\documentclass[12pt]{article}
\usepackage{bm}
\usepackage{amsmath}
\usepackage{amssymb, amscd, amsthm,amsfonts}
\usepackage[dvips]{graphicx}
\usepackage{verbatim}

\newtheorem{theorem}{Theorem}
\newtheorem{lemma}{Lemma}

\newtheorem{corollary}{Corollary}

\newtheorem{rmk}{Remark}
\usepackage{pict2e}

\begin{document}

\title{On the degradations of Binary-Input Discrete Memoryless Channels\footnote{This work was supported by the National Natural Science Foundation of China (No. 61977056).}}
\author{Yadong Jiao, Xiaoyan Cheng, Yuansheng Tang\footnote{Corresponding author.}\\
{\it\small School of Mathematical Sciences, Yangzhou University, Jiangsu, China}\\
and Ming Xu\footnote{Email addresses: dx120210046@stu.yzu.edu.cn(Y. Jiao), xycheng@yzu.edu.cn(X. Cheng), ystang@yzu.edu.cn(Y. Tang), mxu@szcu.edu.cn(M. Xu)}
\\
{\it \small Suzhou City University, Jiangsu, China}}
\date{}
\maketitle

\begin{abstract}
For the polar codes introduced by Arikan in 2009, the first code family achieving the capacity of binary-input discrete memoryless channels (BIDMCs) with low-complexity encoding and decoding, it is crucial to evaluate the reliability of the synthetic channels resulted in the code construction. Since the synthetic channels have an output alphabet that grows exponentially with the code length, an effective method for faithfully evaluating their reliability is to replace them with degradations of manageable alphabet size. The main aim of this paper is to find the optimal degradations of symmetric BIDMCs. We determine all the degradations with the minimum probability of error decoding and give a necessary condition for the degradations with the maximum symmetric capacity. Finally, based on this necessary condition, we propose an efficient algorithm for finding degradation schemes that maximize the symmetric capacity.
\end{abstract}
\vspace{1ex}

\section{Introduction}
In various digital signal processing scenarios, quantization techniques are commonly employed to map large alphabets to smaller ones, serving as a critical means to achieve efficient signal representation and processing. In communication systems, quantizers are core components of receiver design, effectively controlling algorithmic complexity and hardware resource consumption, with their performance directly determining the achievable transmission rate and overall reliability of the system. Thus, designing high-performance channel quantizers has long been an important research topic in the field of communications. Notably, the output of a quantizer shares inherent similarities with the decision results of decision trees, and tree-based classification methods (including classification trees, regression trees, and decision trees) are widely applied in machine learning and pattern recognition \cite{Chou91}. A typical example is handwritten character recognition, where features of handwritten characters are extracted to estimate their corresponding letters and digits. Theoretically, the discrete channel quantization problem is also closely related to classification problems in statistical learning theory \cite{Bishop06}.

Another recent important application of quantization techniques is closely tied to polar codes \cite{Arikan09}. The construction of polar codes is equivalent to evaluating the decoding error probability of each synthetic channel in a set of synthesized channels, but direct evaluation is infeasible due to the extremely large size of the output alphabets of these synthetic channels. One effective approach to address this issue is to degrade or upgrade the synthetic channels to be evaluated into channels with controllable output alphabet sizes \cite{Vardy13}. To the best of the authors' knowledge, this method is currently the only tunable and deterministic solution. In particular, in wiretap channels, degrading and upgrading techniques play a crucial role in ensuring the security and reliability of communication \cite{Shamai10, Skoglund10, Vardy11, Gamal12, Vardy2013}.

The degradation algorithm proposed in \cite{Vardy13} is a greedy-based suboptimal quantization algorithm for binary-input discrete memoryless channels (BIDMCs). When the size of the output alphabet exceeds a predefined threshold, two output symbols are selected and merged (rendering them indistinguishable), with the same operation performed on their symmetric conjugate symbols. The selection of symbols to merge is based on minimizing the loss of mutual information between the channel input and output under the assumption of a uniform input distribution. This merging process reduces the output alphabet size by two, and the procedure is repeated until the alphabet size reaches the predefined threshold. Reference \cite{Pedarsani11} provides an error analysis and approximation guarantees for the channel degradation method described in \cite{Vardy13}. Recently, \cite{Kartowsky19} further proved the optimality of this greedy algorithm in the sense of power-law. Reference \cite{Yagi14} proposes an optimal degradation algorithm for BIDMCs using dynamic programming, which can find a degraded channel that meets the target output alphabet size requirement and maximizes the mutual information between the input and output of this channel. Furthermore, \cite{Ozawa14} applies the SMAWK algorithm \cite{Wilber87} to the dynamic programming algorithm proposed in \cite{Yagi14} by proving the Monge property of mutual information, significantly reducing the computational complexity of this optimal degradation algorithm.

In the aforementioned studies, mutual information is chosen as the objective function for BIDMC degradation, primarily because mutual information represents the theoretical upper bound on the achievable coding rate of a channel \cite{Shannon48}. For different objective functions (such as error exponents, $\alpha$-mutual information, etc.), \cite{Kurkoski12} and \cite{Iwata17} also propose corresponding optimal degradation algorithms for BIDMCs based on dynamic programming. Related studies on the degradation of general discrete memoryless channels \cite{Sharov12, Tal17, Barg18, Kartowsky19, Ordentlich21} can be regarded as generalizations and extensions of the channel degradation methods in \cite{Vardy13} and \cite{Yagi14}.

In this paper, we conduct the following research work. First, we provide a complete characterization of all symmetric degradations of an arbitrary symmetric BIDMC. Second, we identify all P-degradations with the minimum decoding error probability in an arbitrary symmetric BIDMC. Third, we prove that all P-degradations can be viewed as degradations of P$^{\ast}$-degradations composed of segments of the original channel. Finally, we prove that any C-optimal degradation with the maximum symmetric capacity must be a special form of P$^{\ast}$-degradations, namely P$^{+}$-degradations, and derive a necessary condition for this property. Using this condition, we develop an efficient algorithm for maximum symmetric capacity degradations that is superior to that in \cite{Ozawa14}.

The structure of the paper is outlined as follows: In Section~\ref{sec02}, we present key definitions and fundamental properties of BIDMCs, including the likelihood ratio profile, random switching operations, channel equivalence, symmetry, and the concept of channel degradation. In Section~\ref{sec03}, we present a full characterization of all symmetric degradations for an arbitrary symmetric BIDMC, and prove that this degradation relationship remains invariant when an identical channel component is either removed or added. In Section~\ref{sec04}, we identify all P-degradations (i.e., degradations that minimize decoding error probability) for any symmetric BIDMC, and derive a concise criterion for recognizing (4-4)-P-degradations. In Section~\ref{sec05}, we prove that any P-degradation is a degradation of a P$^{\ast}$-degradation, which is composed of components formed from segments of the original channel. We further show that at least three splitting patterns must be adjusted to obtain an upgraded P$^{\ast}$-degradation. In Section~\ref{sec06}, we demonstrate that any C-optimal degradation (i.e., degradations that maximize symmetric capacity) must be a P$^{+}$-degradation that is a P$^{\ast}$-degradation of a special form, and provide a necessary condition for this property. In Section~\ref{sec07}, we propose
an algorithm to generate all candidate C-optimal degradations. Simulation results confirm that the method introduced in Section~\ref{sec06} efficiently yields high-quality estimated channels for Arikan transformations. Finally, we develop an efficient algorithm to find degradation schemes that achieve the maximum symmetric capacity.

\section{Preliminaries}
\label{sec02}
In this paper, for any random variable $v$ the notation $v$ may also express, a concrete value in its sample space, or the probabilistic event that $v$ takes a concrete value, if there is no confusion. For any real number $\varepsilon\in[0,1]$, let $\overline{\varepsilon}$ denote $1-\varepsilon$. For any positive integer $n$, let $[n]$ denote the set $\{1,2,\ldots,n\}$ of integers.

Let $W:x\in\mathcal{X}\mapsto y\in\mathcal{Y}$ be a \emph{binary-input discrete memoryless channel} (BIDMC), where the input $x$ is always supposed to be a random variable {\it uniformly distributed} in $\mathcal{X}=\mathbb{F}_2=\{0,1\}$, the finite field of two elements, and the output alphabet $\mathcal{Y}$ is a discrete set. The \emph{maximum likelihood decoding} (MLD) of the BIDMC $W$ decodes $\hat{y}\in\mathcal{Y}$ into 0 if $\mathcal{L}_W(\hat{y})\geq1$ and 1 otherwise, where $\mathcal{L}_W(\hat{y})=\Pr(y=\hat{y}|x=0)/\Pr(y=\hat{y}|x=1)$ is the \emph{likelihood ratio} (LR) of $\hat{y}$. The {\it decoding error probability} for MLD of $W$ is given by
\begin{equation}\label{01}
	P_{\epsilon}(W)=\frac{1}{2}\sum_{y\in\mathcal{Y}}\min\{\Pr(y|x=0),\Pr(y|x=1)\}.
\end{equation}
The {\it symmetric capacity} of $W$ is defined as
\begin{equation}\label{02} I(W)=\sum_{y\in\mathcal{Y}}\sum_{x\in\mathcal{X}}\frac{1}{2}\Pr(y|x)
\log_2\frac{\Pr(y|x)}{\frac{1}{2}\Pr(y|x=0)+\frac{1}{2}\Pr(y|x=1)},
\end{equation}
which is equal to the {\it mutual information}
$I(x;y)=H(x)+H(y)-H(x,y)$
between the output $y$ and the input $x$.

For $\varepsilon\in[0,1]$, let $P_W(\varepsilon)$ denote the probability $\Pr(y\in L_{W}(\varepsilon))$, where
\begin{align}\label{03}
L_W(\varepsilon)=\{\hat{y}\in\mathcal{Y}: \mathcal{L}_W(\hat{y})=\varepsilon/\overline{\varepsilon}\}.
\end{align}
Since $\{L_W(\varepsilon)\}_{\varepsilon\in[0,1]}$ forms a partition of the output alphabet $\mathcal{Y}$ that is distinguished by LRs, the function $P_W(\varepsilon)$ defined on $[0,1]$ is referred to as the {\it LR-profile} of $W$. Notably, analogous to the Blackwell measure \cite{GR20} and the L/D/G-densities \cite{Urbanke08}, the LR-profile can characterize BIDMCs as LR constitutes a sufficient statistic with respect to decoding \cite{Urbanke08}. Two BIDMCs $W$ and $W'$ are said to be \emph{equivalent}, and denoted $W\cong W'$, if their LR-profiles coincide, i.e.,
$P_{W}(\varepsilon)=P_{W'}(\varepsilon)$ for all $\varepsilon\in[0,1]$.

For any BIDMC $W:x\in \mathcal{X}\mapsto y\in\mathcal{Y}$, another
BIDMC $W':x'\in \mathcal{X}\mapsto y'\in\mathcal{Y}'$ is referred to as a \emph{degradation channel} of $W$, and written as $W'\preccurlyeq W$, if there is a channel $Q:\mathcal{Y}\rightarrow \mathcal{Y}'$ such that
\begin{align}\label{04}
\Pr(y'|x'=a)=\sum_{y\in\mathcal{Y}}\Pr(y|x=a)Q(y'|y),
\end{align}
where $Q(b'|b)$ is the probability of that $Q$ transmits $b\in\mathcal{Y}$ into $b'\in\mathcal{Y}'$. $Q$ is referred to as the \emph{intermediate channel} of the degradation, and $W$ is also referred to as an \emph{upgradation channel} of $W'$. Clearly, the degradation of BIDMCs has {\it transitivity}. Furthermore, BIDMCs can also be distinguished by degradation (c.f. \cite{Blackwell1953,GR20,JCT24}, that is,
\begin{align}\label{05}
W\preccurlyeq W'\preccurlyeq W \text{\ if and only if\ } W\cong W'.
\end{align}

Assume that $\{\mathcal{Y}_1,\ldots,\mathcal{Y}_n\}$ is a partition of the output alphabet $\mathcal{Y}$ of the BIDMC $W$ such that, for each $j\in[n]$, the probability $\Pr(y\in\mathcal{Y}_j| x=a)=q_j$ is independent of $a\in\mathcal{X}$. For $j\in[n]$, let $W_j:x_{j}\in\mathcal{X}\mapsto y_j\in\mathcal{Y}_j$ denote the synthetic BIDMC with transition probabilities
\begin{equation*}
	\Pr(y_j=b| x_{j}=a)=\Pr(y=b| x=a)/q_{j}, \text{ for any }b\in\mathcal{Y}_j\text{ and }a\in\mathcal{X}.
\end{equation*}
Since the input $x$ of $W$ can be regarded as being randomly transmitted over the sub-channels $W_1,\ldots,W_n$, with transmission over $W_j$ occurring with probability $q_j$ for each $j\in[n]$, $W$ is termed a \emph{random switching channel} (RSC) of $W_1,\ldots,W_n$ (c.f. \cite{Tang21}), and is formally expressed as
\begin{align}\label{06}
W=\sum_{j\in[n]}q_jW_j.
\end{align}
Clearly, we have
\begin{align}\label{07}
P_W(\varepsilon)=\sum_{j\in[n]}q_jP_{W_j}(\varepsilon),\ \varepsilon\in[0,1].
\end{align}
Furthermore, if $W'=\sum_{j\in[n]}q'_jW'_j$ is a BIDMC independent of $W$, where $W'_j$ is equivalent to $W_{j}$ for each $j\in[n]$, then for any $p\in[0,1]$ we have (c.f. \cite{JCT24})
\begin{align}\label{08}
pW+\overline{p}W'\cong\sum_{j\in[n]}pq_jW_j+\sum_{j\in[n]}\overline{p}q'_jW'_j
\cong\sum_{j\in[n]}(pq_j+\overline{p}q'_j)W_j,
\end{align}
which indicates that the notation (\ref{06}) admits some natural operations.

For any function $f:[0,1]\rightarrow \mathbb{R}$, let $I_{f}$ be the functional on the collection of BIDMCs defined by
\begin{align}\label{09}
I_{f}(W)=\mathbf{E}(f(\mathcal{L}_{W}(y)/(1+\mathcal{L}_{W}(y))))=\sum_{\varepsilon\in[0,1]}f(\varepsilon)P_{W}(\varepsilon),
\end{align}
where $y$ is the output of $W$. Clearly, for any RSC $W=\sum_{j\in[n]}q_jW_j$ we have
\begin{align}\label{10}
I_{f}(W)=\sum_{j\in[n]}q_jI_{f}(W_{j}).
\end{align}
Therefore, for any BIDMC $W$ and functions $f_{c}(\varepsilon)=1-\hbar(\varepsilon)$ and $f_{e}(\varepsilon)=\min\{\varepsilon, \overline{\varepsilon}\}$, where $\hbar(\varepsilon)=-\varepsilon\log_{2}\varepsilon-\overline{\varepsilon}\log_{2}\overline{\varepsilon}$ is the binary entropy function, we have (c.f. \cite{Urbanke08,GR20})
\begin{align}\label{11}
I(W)=I_{f_{c}}(W),\ P_{\epsilon}(W)=I_{f_{e}}(W).
\end{align}
Therefore, according to (\ref{05}) and the convexity of $f_{c}$, $f_{e}$ one can deduce easily (c.f. \cite{JCT24}) the following well-known results
\begin{align}\label{12}
I(W_{0})\leq I(W_{1}),\ P_{\epsilon}(W_{0})\geq P_{\epsilon}(W_{1}), \text{ if } W_{0}\preccurlyeq W_{1}.
\end{align}

A BIDMC $W$ is said \emph{symmetric} (c.f. \cite{JCT24}) if its LR-profile is symmetric with respect to $1/2$, i.e., $P_W(\varepsilon)=P_W(\overline{\varepsilon}),\ \text{for }\varepsilon\in[0,1]$. For $\varepsilon\in[0,1]$, let $\mathrm{B}(\varepsilon)$ denote the BSC with crossover probability $ \varepsilon $. It is clear that $\overline{p}\mathrm{B}(0)+p\mathrm{B}(1/2)$ is equivalent to the BEC $\mathrm{E}(p)$ with erasure probability $p$, where $\mathrm{B}(0)$ and $\mathrm{B}(1/2)$ are the noiseless channel and the completely noisy channel respectively. Clearly, the BIDMC $W$ is symmetric if and only if it is equivalent to an RSC of some BSCs. For $n\geq 1$, let $\mathbb{B}_n$ denote the set of symmetric BIDMCs that are equivalents of RSCs of $n$ BSCs.

For $n>1$, let $\mathbb{B}^*_n=\mathbb{B}_n\setminus\mathbb{B}_{n-1}$.
Clearly, BIDMC $W$ belongs to $\mathbb{B}^*_n$ if and only if
$W\cong\sum_{i\in[n]}p_i\mathrm{B}(\varepsilon_i)$ for some positive numbers $p_1,\ldots,p_{n}$ with sum 1 and $0\leq \varepsilon_1<\cdots<\varepsilon_{n}\leq 1/2$, especially its LR-profile is given by
\begin{align}\label{13}
 P_W(\varepsilon)=\left\{
 \begin{array}{ll}
 p_{i}/2, & \text{if }\varepsilon=\varepsilon_{i}\neq 1/2, \\
 p_{n}, & \text{if }\varepsilon=\varepsilon_{n}=1/2, \\
 0,&\text{otherwise},
 \end{array}
\right.
\end{align}
and therefore we have
\begin{align}\label{14}
I(W)=\sum_{i\in[n]}p_{i}f_{c}(\varepsilon_{i})=1-\sum_{i\in[n]}p_{i}\hbar(\varepsilon_{i}),
\end{align}
\begin{align}\label{15}
P_{\epsilon}(W)=\sum_{i\in[n]}p_{i}\varepsilon_{i}.
\end{align}

\section{Degradations of Symmetric BIDMCs}
\label{sec03}
In this section, we mainly deal with the symmetric degradations of symmetric BIDMCs. For distributions $q_{1},\ldots,q_{m}$ and $p_{1},\ldots,p_{n}$, a matrix $(k_{i,j})_{m\times n}$ is referred to as a \emph{1-matrix} of \emph{pattern} $(q_{1},\ldots,q_{m};p_{1},\ldots,p_{n})$ if $k_{i,j}\geq 0$ and
\begin{align}
\sum_{j\in[n]}k_{i,j}=q_{i}, i\in[m],\label{16}\\
\sum_{i\in[m]}k_{i,j}=p_{j}, j\in[n].\label{17}
\end{align}
\begin{lemma}\label{lem01}
If $Q_{1},\ldots,Q_{n}$ are symmetric BIDMCs with
\begin{align*}
\sum_{j\in[n]}p_{j}Q_{j}\cong\sum_{i\in[m]}q_{i}\mathrm{B}(\sigma_{i})\in \mathbb{B}^{\ast}_{m},
\end{align*}
where $p_{1},\ldots,p_{n}$ are positive numbers with sum 1, then there is a 1-matrix $(k_{i,j})_{m\times n}$ of pattern $(q_{1},\ldots,q_{m};p_{1},\ldots,p_{n})$ such that
\begin{align}\label{18}
Q_{j}\cong\sum_{i\in[m]}\frac{k_{i,j}}{p_{j}}\mathrm{B}(\sigma_{i}),\ j\in[n].
\end{align}
\end{lemma}
\begin{proof}
This lemma can be proved simply by investigating the LR-profiles of the BIDMCs $Q_{1},\ldots,Q_{n}$.
\end{proof}

For $W=\sum_{j\in[n]}p_{j}\mathrm{B}(\varepsilon_{j})\in\mathbb{B}_{n}, i\in[n]$ and $0\leq p'_{i}\leq p_{i}, p'_{i}\mathrm{B}(\varepsilon_{i})$ is referred to as a \emph{particle} of $W$ of \emph{weight} $p'_{i}$. The particle $p\mathrm{B}(\varepsilon)$ is referred to as a \emph{degradation} of particles $k_{1}\mathrm{B}(\sigma_{1}),\ldots,k_{m}\mathrm{B}(\sigma_{m})$ if $p=\sum_{i\in[m]}k_{i}\leq1$ and $\mathrm{B}(\varepsilon)$ is a degradation of the BIDMC $\sum_{i\in[m]}\frac{k_{i}}{p}\mathrm{B}(\sigma_{i})$ provided $p>0$. The following theorem gives a sufficient and necessary condition for symmetric degradations of symmetric BIDMCs.

\begin{theorem}\label{the01}
Let $W=\sum_{j\in[n]}p_{j}\mathrm{B}(\varepsilon_{j})\in\mathbb{B}_{n}$ and $Q=\sum_{i\in[m]}q_{i}\mathrm{B}(\sigma_{i})\in\mathbb{B}_{m}$ be BIDMCs with $\varepsilon_{1},\ldots,\varepsilon_{n},\sigma_{1},\ldots,\sigma_{m}\in[0,1/2]$. Then $W$ is a degradation of $Q$ if and only if there is a 1-matrix $(k_{i,j})_{m\times n}$ of pattern $(q_{1},\ldots,q_{m};p_{1},\ldots,p_{n})$ such that
\begin{align}\label{19}
p_{j}\varepsilon_{j}\geq\sum_{i\in[m]}k_{i,j}\sigma_{i},\ j\in[n],
\end{align}
that is, for each $i\in[m]$ the particle $q_{i}\mathrm{B}(\sigma_{i})$ can be split into particles $k_{i,1}\mathrm{B}(\sigma_{i}),\ldots,k_{i,n}\mathrm{B}(\sigma_{i})$ such that for each $j\in[n]$ the particle $p_{j}\mathrm{B}(\varepsilon_{j})$ is a degradation of the particles $k_{1,j}\mathrm{B}(\sigma_{1}),\ldots,k_{m,j}\mathrm{B}(\sigma_{m})$.
\end{theorem}
\begin{proof}
The proof is given in Appendix A.
\end{proof}

Notice that this theorem implies that, for $\varepsilon, \sigma_{1},\ldots,\sigma_{m}\in[0,1/2]$, the BSC $\mathrm{B}(\varepsilon)$ is a degradation of $\sum_{i\in[m]}q_{i}\mathrm{B}(\sigma_{i})\in\mathbb{B}_{m}$ if and only if $\varepsilon\geq\sum_{i\in[m]}q_{i}\sigma_{i}$.

According to Theorem \ref{the01}, one can derive the following corollary simply.

\begin{corollary}\label{cor01}
Assume that $W, W_{1},\ldots,W_{s}$ are symmetric BIDMCs with $W\cong\sum_{l\in[s]}r_{l}W_{l}$, where $r_{l}>0$ for each $l\in[s]$. Then, $W$ is a degradation of symmetric BIDMC $Q$ if and only if there are symmetric BIDMCs $Q_{1},\ldots,Q_{s}$ such that $Q=\sum_{l\in[s]}r_{l}Q_{l}$ and $W_{l}\preccurlyeq Q_{l}$ for each $l\in[s]$.
\end{corollary}
\begin{proof}
Let $W=\sum_{j\in[n]}p_{j}\mathrm{B}(\varepsilon_{j})\in\mathbb{B}^{\ast}_{n}$ and $Q=\sum_{i\in[m]}q_{i}\mathrm{B}(\sigma_{i})\in\mathbb{B}^{\ast}_{m}$, where $\varepsilon_{1},\ldots,\varepsilon_{n},\sigma_{1},\ldots,\sigma_{m}\in[0,1/2]$. According to Theorem \ref{the01}, $W$ is a degradation of $Q$ if and only if there is a 1-matrix $(k_{i,j})_{m\times n}$ of pattern $(q_{1},\ldots,q_{m};p_{1},\ldots,p_{n})$ with (\ref{19}). From Lemma \ref{lem01} and $W\cong\sum_{l\in[s]}r_{l}W_{l}\cong\sum_{j\in[n]}p_{j}\mathrm{B}(\varepsilon_{j})\in\mathbb{B}^{\ast}_{n}$, there is a 1-matrix $(t_{j,l})_{n\times s}$ of pattern $(p_{1},\ldots,p_{n};r_{1},\ldots,r_{s})$ such that
\begin{align}\label{20}
W_{l}=\sum_{j=1}^{n}\frac{t_{j,l}}{r_{l}}\mathrm{B}(\varepsilon_{j}),\ l\in[s].
\end{align}
For $l\in[s]$, let
\begin{align}\label{21}
Q_{l}=\sum_{j=1}^{n}\sum_{i=1}^{m}h_{i,j}^{(l)}\mathrm{B}(\sigma_{i}),
\end{align}
where $h_{i,j}^{(l)}=(t_{j,l}k_{i,j})/(r_{l}p_{j}), i\in[m], j\in[n]$. Clearly, we have
\begin{align*}
\sum_{l=1}^{s}r_{l}Q_{l}\cong\sum_{l=1}^{s}\sum_{j=1}^{n}\sum_{i=1}^{m}\frac{t_{j,l}k_{i,j}}{p_{j}}\mathrm{B}(\sigma_{i})
\cong\sum_{j=1}^{n}\sum_{i=1}^{m}k_{i,j}\mathrm{B}(\sigma_{i})\cong\sum_{i=1}^{m}q_{i}\mathrm{B}(\sigma_{i})\cong Q.
\end{align*}
From (\ref{19}) and $\sum_{i=1}^{m}h_{i,j}^{(l)}=t_{j,l}/r_{l}$ we see
\begin{align*}
\mathrm{B}(\varepsilon_{j})\preccurlyeq\sum_{i=1}^{m}\frac{1}{p_{j}}k_{i,j}\mathrm{B}(\sigma_{i})=\sum_{i=1}^{m}\frac{1}{t_{j,l}/r_{l}}h_{i,j}^{(l)}\mathrm{B}(\sigma_{i})
\text{ if } t_{j,l}>0,
\end{align*}
and then from (\ref{20}) and (\ref{21}) we see $W_{l}\preccurlyeq Q_{l}, l\in[s]$.
\end{proof}

Furthermore, we have the following lemma.

\begin{lemma}\label{lem02}
Let $W, Q$ be symmetric BIDMCs such that $p\mathrm{B}(\varepsilon)+\overline{p}W\preccurlyeq p\mathrm{B}(\varepsilon)+\overline{p}Q$ for some $p\in(0,1)$ and $\varepsilon\in[0,1]$. Then, we must have $W\preccurlyeq Q$.
\end{lemma}
\begin{proof}
According to Corollary \ref{cor01}, we see that there are a number $a$ with $0\leq a\leq p$ and symmetric BIDMCs $Q_{0}, Q_{1}$ such that
\begin{align*}
&\mathrm{B}(\varepsilon)\preccurlyeq \frac{a}{p}\mathrm{B}(\varepsilon)+\frac{p-a}{p}Q_{0},\\
W&\preccurlyeq \frac{p-a}{\overline{p}}\mathrm{B}(\varepsilon)+\frac{\overline{p}-p+a}{\overline{p}}Q_{1}
\end{align*}
and $Q\cong\frac{p-a}{\overline{p}}Q_{0}+\frac{\overline{p}-p+a}{\overline{p}}Q_{1}$. If $a=p$, then we have $W\preccurlyeq Q_{1}\cong Q$. Now we assume $a<p$. By induction one can show easily that, for any positive integer $k$,
\begin{align}\label{22}
W\preccurlyeq\frac{(p-a)a^{k}}{\overline{p}p^{k}}\mathrm{B}(\varepsilon)+\frac{\overline{p}-p+a}{\overline{p}}Q_{1}+\frac{(p-a)^{2}}{\overline{p}p}\sum_{i\in[k]}\frac{a^{i-1}}{p^{i-1}}Q_{0}.
\end{align}
As $k$ approaches to infinity, the limitation of RHS of (\ref{22}) is equivalent to $\frac{\overline{p}-p+a}{\overline{p}}Q_{1}+\frac{p-a}{\overline{p}}Q_{0}\cong Q$. Therefore, we also have $W\preccurlyeq Q$.
\end{proof}

According to this lemma, by induction one can deduce simply the following corollary, which implies that degradation of BIDMCs is preserved under the removal or addition of identical components.

\begin{corollary}\label{cor02}
Let $R, W, Q$ be symmetric BIDMCs such that $pR + \overline{p}W \preccurlyeq pR + \overline{p}Q$ for some $p\in(0,1)$. Then, we must have $W\preccurlyeq Q$.
\end{corollary}

\section{Degradations with the Lowest Decoding Error Probability}
\label{sec04}
We say $W\in \mathbb{B}_{n}$ is a $(2m,2n)$\emph{-P-degradation} \emph{(or $2n$-P-degradation, or P-degradation)} of $Q\in \mathbb{B}_{m}$, and write $W\preccurlyeq_{\text{P}} Q$, if $W$ is a degradation of $Q$ with the lowest decoding error probability:
\begin{align}\label{23}
P_{\epsilon}(W)=\min\{P_{\epsilon}(W'): W'\in \mathbb{B}_{n}, W'\preccurlyeq Q\}.
\end{align}
\begin{lemma}\label{lem03}
Any symmetric BIDMC has a unique 2-P-degradation. Indeed, the unique 2-P-degradation of $\sum_{i\in[m]}q_{i}\mathrm{B}(\sigma_{i})\in \mathbb{B}_{m}$ with $\sigma_{1},\ldots,\sigma_{m}\in[0,1/2]$ is the BSC $\mathrm{B}(\varepsilon)$ with $\varepsilon=\sum_{i\in[m]}q_{i}\sigma_{i}$, which is the weighted average of the
cross probabilities $\sigma_{1},\ldots,\sigma_{m}$.
\end{lemma}
\begin{proof}
Let $\varepsilon$ be a number in $[0,1/2]$ such that $\mathrm{B}(\varepsilon)$ is a degradation of $Q=\sum_{i\in[m]}q_{i}\mathrm{B}(\sigma_{i})\in \mathbb{B}_{m}^{\ast}$. According to (\ref{44}) there are numbers $e_{1},\ldots,e_{m}$ in $[0,1]$ such that
\begin{align*}
\varepsilon=\sum_{i\in[m]}q_{i}(\sigma_{i}\overline{e_{i}}+\overline{\sigma_{i}}e_{i}).
\end{align*}
Hence $P_{\epsilon}(\mathrm{B}(\varepsilon))=\varepsilon$ attains the minimum if and only if the numbers $e_{1},\ldots,e_{m}$ are equal to 0 and consequently $\varepsilon=\sum_{i\in[m]}q_{i}\sigma_{i}$.
\end{proof}

We say $p\mathrm{B}(\varepsilon)$ is \emph{compounded} from the particles $k_{1}\mathrm{B}(\sigma_{1}),\ldots,k_{m}\mathrm{B}(\sigma_{m})$ if $p=\sum_{i\in[m]}k_{i}$, $\{\sigma_{1},\ldots,\sigma_{m}\}\subset[0, 1/2]$ and $p\varepsilon=\sum_{i\in[m]}k_{i}\sigma_{i}$, i.e., $\mathrm{B}(\varepsilon)$ is the 2-P-degradation of $\sum_{i\in[m]}\frac{k_{i}}{p}\mathrm{B}(\sigma_{i})\in \mathbb{B}_{m}$ provided $p>0$.

\begin{theorem}\label{the02}
Assume $W=\sum_{j\in[n]}p_{j}\mathrm{B}(\varepsilon_{j})\in \mathbb{B}_{n}^{\ast}$ is a degradation of $Q=\sum_{i\in[m]}q_{i}\mathrm{B}(\sigma_{i})\in \mathbb{B}_{m}$, where $\varepsilon_{1},\ldots,\varepsilon_{n},\sigma_{1},\ldots,\sigma_{m}\in[0,1/2]$. Then, $W$ is a $2n$-P-degradation of $Q$ if and only if there is a 1-matrix $(k_{i,j})_{m\times n}$ of pattern $(q_{1},\ldots,q_{m};p_{1},\ldots,p_{n})$ such that
\begin{align}\label{24}
\sum_{i\in[m]}k_{i, j}\sigma_{i}=p_{j}\varepsilon_{j},\ j\in[n],
\end{align}
that is, for each $i\in[m]$ the particle $q_{i}\mathrm{B}(\sigma_{i})$ can be split into particles $k_{i, 1}\mathrm{B}(\sigma_{i}),\ldots,k_{i, n}\mathrm{B}(\sigma_{i})$ such that for each $j\in[n]$ the particle $p_{j}\mathrm{B}(\varepsilon_{j})$ is compounded from the particles $k_{1, j}\mathrm{B}(\sigma_{1}),\ldots,k_{m, j}\mathrm{B}(\sigma_{m})$.
\end{theorem}
\begin{proof}
According to Theorem \ref{the01}, there is a 1-matrix $(k_{i,j})_{m\times n}$ of pattern $(q_{1},\ldots,q_{m};p_{1},\ldots,p_{n})$ such that $\mathrm{B}(\varepsilon_{j})\preccurlyeq\sum_{i\in[m]}\frac{k_{i,j}}{p_{j}}\mathrm{B}(\sigma_{i})$ for each $j\in[n]$. Clearly, we have $P_{\epsilon}(W)=\sum_{j\in[n]}p_{j}\varepsilon_{j}\geq \sum_{i\in[m]}q_{i}\sigma_{i}=P_{\epsilon}(Q)$.

If for $j\in[n]$ the particle $p_{j}\mathrm{B}(\varepsilon_{j})$ is compounded from the particles $k_{1, j}\mathrm{B}(\sigma_{1}),\ldots,k_{m, j}\mathrm{B}(\sigma_{m})$, according to Lemma \ref{lem03} we have (\ref{24}), and then from
\begin{align*}
P_{\epsilon}(W)=\sum_{j\in[n]}p_{j}\varepsilon_{j}=\sum_{j\in[n]}\sum_{i\in[m]}k_{i,j}\sigma_{i}=\sum_{i\in[m]}q_{i}\sigma_{i}=P_{\epsilon}(Q)
\end{align*}
we see that $W$ is a $2n$-P-degradation of $Q$.

Now we assume $W\preccurlyeq_{\text{P}} Q$. Clearly, for each $j\in[n]$, the BSC $\mathrm{B}(\varepsilon_{j})$ is a 2-P-degradation of $\sum_{i\in[m]}\frac{k_{i,j}}{p_{j}}\mathrm{B}(\sigma_{i})$. Hence, according to Lemma \ref{lem03} we see (\ref{24}) and thus the particle $p_{j}\mathrm{B}(\varepsilon_{j})$ is compounded from the particles $k_{1, j}\mathrm{B}(\sigma_{1}),\ldots,k_{m, j}\mathrm{B}(\sigma_{m})$.
\end{proof}

Although Theorem \ref{the01} (or Theorem \ref{the02}) enables the explicit construction of all symmetric degradations (or P-degradations, respectively) for any symmetric BIDMC, verifying the relation $W\preccurlyeq Q$ or $W\preccurlyeq_{\text{P}} Q$ for given symmetric BIDMCs $W$ and $Q$ is far from straightforward. The following lemma gives some necessary conditions for such verifications.

\begin{lemma}\label{lem04}
Assume $W=\sum_{j\in[n]}p_{j}\mathrm{B}(\varepsilon_{j})\in \mathbb{B}_{n}^{\ast}$ is a degradation of $Q=\sum_{i\in[m]}q_{i}\mathrm{B}(\sigma_{i})\in \mathbb{B}_{m}^{\ast}$, where $0\leq \varepsilon_{1}<\cdots<\varepsilon_{n}\leq1/2$ and $0\leq \sigma_{1}<\cdots<\sigma_{m}\leq1/2$. Then, we have $\varepsilon_{1}\geq\sigma_{1}$. Furthermore, if $W$ is a $2n$-P-degradation of $Q$, then we have $\varepsilon_{n}\leq\sigma_{m}$.
\end{lemma}
\begin{proof}
According to Theorem \ref{the01}, there is a 1-matrix $(k_{i, j})_{m\times n}$ of pattern $(q_{1},\ldots,q_{m};p_{1},\ldots,p_{n})$ such that $\mathrm{B}(\varepsilon_{j})\preccurlyeq Q_{j}=\sum_{i\in[m]}\frac{k_{i,j}}{p_{j}}\mathrm{B}(\sigma_{i})$ for $j\in[n]$. Hence, $\varepsilon_{1}=P_{\epsilon}(\mathrm{B}(\varepsilon_{1}))\geq P_{\epsilon}(Q_{j})=\sum_{i\in[m]}\frac{k_{i,1}}{p_{1}}\sigma_{i}\geq \sigma_{1}$. Furthermore, if $W$ is a $2n$-P-degradation of $Q$, then $\mathrm{B}(\varepsilon_{n})\preccurlyeq \sum_{i\in[m]}\frac{k_{i,n}}{p_{n}}\mathrm{B}(\sigma_{i})$ is a 2-P-degradation and thus $\varepsilon_{n}=P_{\epsilon}(\mathrm{B}(\varepsilon_{n}))=P_{\epsilon}(Q_{n})=\sum_{i\in[m]}\frac{k_{i,n}}{p_{n}}\sigma_{i}\leq \sigma_{m}$.
\end{proof}

We conclude this section by characterizing all the $(4,4)$-P-degradations of BIDMCs.

\begin{lemma}\label{lem05}
For $0\leq \sigma_{1}<\sigma_{2}\leq1/2, 0\leq \varepsilon_{1}<\varepsilon_{2}\leq1/2$ and $Q=\overline{q}\mathrm{B}(\sigma_{1})+q\mathrm{B}(\sigma_{2})\in \mathbb{B}_{2}^{\ast}$, there is a number $p\in(0,1)$ such that $W=\overline{p}\mathrm{B}(\varepsilon_{1})+p\mathrm{B}(\varepsilon_{2})$ is a (4,4)-P-degradation of $Q$ if and only if
\begin{align}\label{25}
\sigma_{1}\leq \varepsilon_{1}<\overline{q}\sigma_{1}+q\sigma_{2}<\varepsilon_{2}\leq\sigma_{2}.
\end{align}
Notice that $p$ is the unique number determined by
\begin{align}\label{26}
\overline{p}\varepsilon_{1}+p\varepsilon_{2}=\overline{q}\sigma_{1}+q\sigma_{2}.
\end{align}
\end{lemma}
\begin{proof}
If $W=\overline{p}\mathrm{B}(\varepsilon_{1})+p\mathrm{B}(\varepsilon_{2})$ is a P-degradation of $Q$ for some $p\in(0,1)$, we have
\begin{equation*}
\overline{p}\varepsilon_{1}+p\varepsilon_{2}=P_{\epsilon}(W)=P_{\epsilon}(Q)=\overline{q}\sigma_{1}+q\sigma_{2}
\end{equation*}
and then according to Lemma \ref{lem04} we see (\ref{25}).

Now we assume (\ref{25}) and $p$ is the number determined by (\ref{26}). Clearly, we have $p\in(0,1)$. Let $a_{j}=(\varepsilon_{j}-\sigma_{1})/(\sigma_{2}-\sigma_{1}),\ j=1, 2$. Clearly, we have $0\leq a_{j}\leq1$ and $\varepsilon_{j}=\overline{a_{j}}\sigma_{1}+a_{j}\sigma_{2}, j=1, 2$. Let $k_{1,1}=\overline{p}\,\overline{a_{1}}, k_{1,2}=p\overline{a_{2}}, k_{2,1}=\overline{p}a_{1}$ and $k_{2,2}=pa_{2}$. Then, according to (\ref{26}) we have
\[
k_{1,1}+k_{2,1}=\overline{p}\,\overline{a_{1}}+\overline{p}a_{1}=\overline{p},\ k_{1,2}+k_{2,2}=p\overline{a_{2}}+pa_{2}=p,
\]
\[
k_{2,1}+k_{2,2}=\overline{p}a_{1}+pa_{2}=\frac{\overline{p}(\varepsilon_{1}-\sigma_{1})}{\sigma_{2}-\sigma_{1}}+\frac{p(\varepsilon_{2}-\sigma_{1})}{\sigma_{2}-\sigma_{1}}
=\frac{\overline{q}\sigma_{1}+q\sigma_{2}-\sigma_{1}}{\sigma_{2}-\sigma_{1}}=q,
\]
\[
k_{1,1}+k_{1,2}=(k_{1,1}+k_{2,1})+(k_{1,2}+k_{2,2})-(k_{2,1}+k_{2,2})=\overline{p}+p-q=\overline{q}.
\]
Therefore, according to Theorem \ref{the02} and
\begin{equation*}
\varepsilon_{j}=\overline{a_{j}}\sigma_{1}+a_{j}\sigma_{2}=(k_{1,j}\sigma_{1}+k_{2,j}\sigma_{2})/(k_{1,j}+k_{2,j}),\ j=1,2
\end{equation*}
we see $W=\overline{p}\mathrm{B}(\varepsilon_{1})+p\mathrm{B}(\varepsilon_{2})\preccurlyeq_{\text{P}} Q=\overline{q}\mathrm{B}(\sigma_{1})+q\mathrm{B}(\sigma_{2})$.
\end{proof}
\section{Optimization of the P-degradations}
\label{sec05}
For any given symmetric BIDMC, according to Theorem \ref{the02} there are infinite $2n$-P-degradations, each of them minimizes the probability of error decoding. To optimize the $2n$-degradations, in this section we consider to derive conditions for that one $2n$-P-degradation may be degradation of another $2n$-P-degradation.
\begin{lemma}\label{lem06}
Assume that $W=\sum_{j\in[n]}p_{j}\mathrm{B}(\varepsilon_{j})\in \mathbb{B}_{n}$ is a P-degradation of $Q=\sum_{i\in[m]}q_{i}\mathrm{B}(\sigma_{i})\in \mathbb{B}_{m}$, where $0\leq \sigma_{1}<\cdots<\sigma_{m}\leq1/2$ and $\sigma_{1}\leq \varepsilon_{1}\leq\varepsilon_{2}\leq\cdots\leq\varepsilon_{n}\leq\sigma_{m}$. Let $(k_{i,j})_{m\times n}$ be the 1-matrix of pattern $(q_{1},\ldots,q_{m};p_{1},\ldots,p_{n})$ with (\ref{24}).

\noindent
1.
If $k_{i', j'}\neq0$ for some $i', j'$ with $\sigma_{i'}\geq \varepsilon_{j'+1}$, let
\begin{gather*}
k'_{i,j}=
\begin{cases}
	0, & \text{ if }(i,j)=(i',j'),\\
    k_{i,j}+k_{i', j'}, & \text{ if }(i,j)=(i',j'+1),\\
	k_{i,j}, & \text{ otherwise },
\end{cases}
\end{gather*}
\begin{equation*}
p'_{j}=\sum_{i\in[m]}k'_{i,j},\  \varepsilon'_{j}=\frac{1}{p'_{j}}\sum_{i\in[m]}k'_{i,j}\sigma_{i}.
\end{equation*}

Then, $W\preccurlyeq_{\text{P}} W'=\sum_{j\in[n]}p'_{j}\mathrm{B}(\varepsilon'_{j})\preccurlyeq_{\text{P}} Q$.

\noindent
2.
If $k_{i'', j''}\neq0$ for some $i'', j''$ with $\sigma_{i''}\leq \varepsilon_{j''-1}$, let
\begin{gather*}
k''_{i,j}=
\begin{cases}
	0, & \text{ if }(i,j)=(i'',j''),\\
    k_{i,j}+k_{i'', j''}, & \text{ if }(i,j)=(i'',j''-1),\\
	k_{i,j}, & \text{ otherwise },
\end{cases}
\end{gather*}
\begin{equation*}
p''_{j}=\sum_{i\in[m]}k''_{i,j},\  \varepsilon''_{j}=\frac{1}{p''_{j}}\sum_{i\in[m]}k''_{i,j}\sigma_{i}.
\end{equation*}

Then, $W\preccurlyeq_{\text{P}} W''=\sum_{j\in[n]}p''_{j}\mathrm{B}(\varepsilon''_{j})\preccurlyeq_{\text{P}} Q$.

\noindent
3.
If $k_{i_{\ast}, j^{\ast}}k_{i^{\ast}, j_{\ast}}\neq0$ for some $i_{\ast}, i^{\ast}, j_{\ast}, j^{\ast}$ with $\varepsilon_{j_{\ast}}\leq \sigma_{i_{\ast}}<\sigma_{i^{\ast}}\leq \varepsilon_{j^{\ast}}$, let
\begin{gather*}
k_{i,j}^{\ast}=
\begin{cases}
	k_{i,j}-\delta, & \text{ if }(i,j)=(i_{\ast},j^{\ast})\text{ or }(i^{\ast},j_{\ast}),\\
    k_{i,j}+\delta, & \text{ if }(i,j)=(i_{\ast},j_{\ast})\text{ or }(i^{\ast},j^{\ast}),\\
	k_{i,j}, & \text{ otherwise },
\end{cases}
\end{gather*}
\begin{equation*}
p_{j}^{\ast}=\sum_{i\in[m]}k_{i,j}^{\ast},\  \varepsilon_{j}^{\ast}=\frac{1}{p_{j}^{\ast}}\sum_{i\in[m]}k_{i,j}^{\ast}\sigma_{i},
\end{equation*}
where $\delta=\min\{k_{i_{\ast},j^{\ast}}, k_{i^{\ast},j_{\ast}}\}$. Then, $W\preccurlyeq_{\text{P}} W^{\ast}=\sum_{j\in[n]}p_{j}^{\ast}\mathrm{B}(\varepsilon_{j}^{\ast})\preccurlyeq_{\text{P}} Q$.
\end{lemma}
 \begin{figure}[t]
\begin{center}
\includegraphics[width=0.6\linewidth]{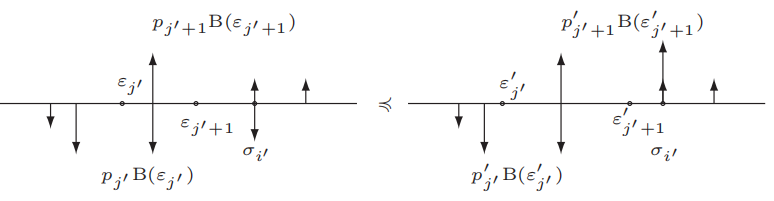}
\caption{{ \small $k_{i',j'}\neq 0$ with $\sigma_{i'}\geq \varepsilon_{j'+1}$}.}
 \label{fig01}
 \end{center}
\end{figure}

 \begin{figure}[t]
\begin{center}
\includegraphics[width=0.6\linewidth]{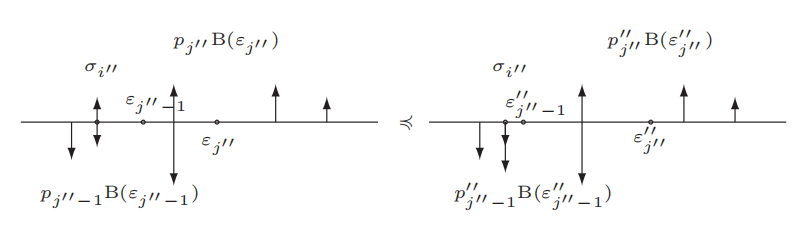}
\caption{{ \small $k_{i'',j''}\neq 0$ with $\sigma_{i''}\leq \varepsilon_{j''-1}$}.}
 \label{fig02}
 \end{center}
\end{figure}

\begin{proof}
The proof is given in Appendix B.
\end{proof}
\begin{rmk}\label{rem01}
The first conclusion in Lemma \ref{lem06} implies: if $k_{i',j'}\neq0$ for some $i', j'$ with $\sigma_{i'}\geq \varepsilon_{j'+1}$, then a better P-degradation may be obtained from $W$ by replacing $p_{j'}\mathrm{B}(\varepsilon_{j'})$ with the particle compounded from those in $\{k_{i,j'}\mathrm{B}(\sigma_{i}): i\in[m]\setminus\{i'\}\}$, and replacing $p_{j'+1}\mathrm{B}(\varepsilon_{j'+1})$ with the particle compounded from those in $\{(k_{i',j'}+k_{i',j'+1})\mathrm{B}(\sigma_{i'})\}\cup\{k_{i,j'+1}\mathrm{B}(\sigma_{i}): i\in[m]\setminus\{i'\}\}$, respectively,
as depicted in Figure \ref{fig01}.

The second conclusion in Lemma \ref{lem06} implies: if $k_{i'',j''}\neq0$ for some $i'', j''$ with $\sigma_{i''}\leq \varepsilon_{j''-1}$, then a better P-degradation may be obtained from $W$ by replacing $p_{j''}\mathrm{B}(\varepsilon_{j''})$ with the particle compounded from those in $\{k_{i,j''}\mathrm{B}(\sigma_{i}): i\in[m]\setminus\{i''\}\}$, and replacing $p_{j''-1}\mathrm{B}(\varepsilon_{j''-1})$ with the particle compounded from those in $\{(k_{i'',j''}+k_{i'',j''-1})\mathrm{B}(\sigma_{i''})\}\cup\{k_{i,j''-1}\mathrm{B}(\sigma_{i}): i\in[m]\setminus\{i''\}\}$, respectively, as depicted in Figure \ref{fig02}.

The third conclusion in Lemma \ref{lem06} implies: if $k_{i_{\ast},j^{\ast}}k_{i^{\ast},j_{\ast}}\neq0$ for some $i_{\ast}, i^{\ast}, j_{\ast}, j^{\ast}$ with $\varepsilon_{j_{\ast}}\leq\sigma_{i_{\ast}}<\sigma_{i^{\ast}}\leq\varepsilon_{j^{\ast}}$, then a better P-degradation may be obtained from $W$ by replacing $p_{j_{\ast}}\mathrm{B}(\varepsilon_{j_{\ast}})$ with the particle compounded from those in
\begin{equation*}
\{(k_{i_{\ast},j_{\ast}}+\delta)\mathrm{B}(\sigma_{i_{\ast}}), (k_{i^{\ast},j_{\ast}}-\delta)\mathrm{B}(\sigma_{i^{\ast}})\}\cup\{k_{i,j_{\ast}}\mathrm{B}(\sigma_{i}): i\in[m]\setminus\{i_{\ast}, i^{\ast}\}\},
\end{equation*}
and replacing $p_{j^{\ast}}\mathrm{B}(\varepsilon_{j^{\ast}})$ with the particle compounded from those in
\begin{equation*}
\{(k_{i_{\ast},j^{\ast}}-\delta)\mathrm{B}(\sigma_{i_{\ast}}), (k_{i^{\ast},j^{\ast}}+\delta)\mathrm{B}(\sigma_{i^{\ast}})\}\cup\{k_{i,j^{\ast}}\mathrm{B}(\sigma_{i}): i\in[m]\setminus\{i_{\ast}, i^{\ast}\}\},
\end{equation*}
respectively, as depicted in Figure \ref{fig03}, where $\delta=\min\{k_{i_{\ast},j^{\ast}},k_{i^{\ast},j_{\ast}}\}$.
\end{rmk}

 \begin{figure}[t]
\begin{center}
\includegraphics[width=0.6\linewidth]{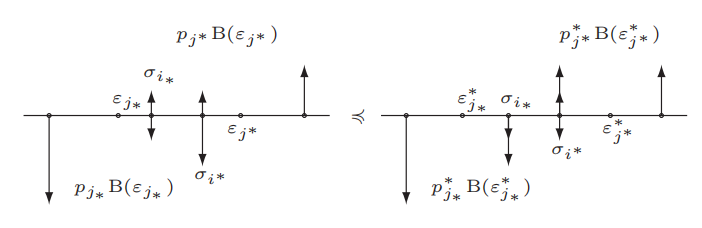}
\caption{{ \small $k_{i_{\ast},j^{\ast}}k_{i^{\ast},j_{\ast}}\neq 0$ with $\varepsilon_{j_{\ast}}\leq \sigma_{i_{\ast}}<\sigma_{i^{\ast}}\leq \varepsilon_{j^{\ast}}$}.}
 \label{fig03}
 \end{center}
\end{figure}

Let $q_{0}=q_{m+1}=0$, $Q=\sum_{i\in[m]}q_{i}\mathrm{B}(\sigma_{i})\in \mathbb{B}_{m}^{\ast}$, where $0\leq\sigma_{1}<\cdots<\sigma_{m}\leq1/2$. For $0\leq j<j'\leq m+1$, $0\leq t\leq q_{j}$ and $0\leq t'\leq q_{j'}$, $S=\sum_{j\leq i\leq j'}r_{i}\mathrm{B}(\sigma_{i})$ is referred to as the $(j,j';t,t')$-\emph{segment} of $Q$ if $r_{j}=t$, $r_{j'}=q_{j'}-t'$ and $r_{i}=q_{i}$ for $j<i<j'$. Note that, the $(j,j';q_{j},t')$- and $(j-1,j';0,t')$-segments coincide, and so do the $(j,j';t,0)$- and $(j,j'+1;0,q_{j'+1})$-segments. For $1\leq i_{2}<\cdots<i_{n}\leq m$ and $0\leq s_{j}\leq q_{i_{j}}$, $j=2,3,\ldots,n$, let $D_{Q}(i_{2},\ldots,i_{n};s_{2},\ldots,s_{n})=\sum_{j\in[n]}p_{j}\mathrm{B}(\varepsilon_{j})$ be the BIDMC in $\mathbb{B}_{n}^{\ast}$ such that for each $j\in[n]$ the particle $p_{j}\mathrm{B}(\varepsilon_{j})$ is \emph{compounded} from the $(i_{j},i_{j+1};s_{i_{j}},s_{i_{j+1}})$-segment of $Q$, i.e., $p_{j}\mathrm{B}(\varepsilon_{j})$ is compounded from the particles in
\begin{align}\label{27}
\{s_{j}\mathrm{B}(\sigma_{i_{j}}),(q_{i_{j+1}}-s_{j+1})\mathrm{B}(\sigma_{i_{j+1}})\}\cup\{q_{i}\mathrm{B}(\sigma_{i}): i_{j}<i<i_{j+1}\},
\end{align}
where $i_{1}=s_{1}=0$, $i_{n+1}=m+1$ and $s_{n+1}=0$. The pairs $(i_{2},s_{2}),\ldots,(i_{n},s_{n})$ are referred to as the \emph{splitting patterns} of $D_{Q}(i_{2},\ldots,i_{n};s_{2},\ldots,s_{n})$. Clearly, any splitting pattern $(i_{l},s_{l})$ with $s_{l}=q_{i_{l}}$ can be equivalently replaced by $(i_{l}-1,0)$ if $i_{l-1}<i_{l}-1$, and any splitting pattern $(i_{l},s_{l})$ with $s_{l}=0$ can be equivalently replaced by $(i_{l}+1,q_{i_{l}+1})$ if $i_{l+1}>i_{l}+1$.

According to Lemma \ref{lem06}, one can deduce the following theorem easily.

\begin{theorem}\label{the03}
Let $W$ be a $2n$-degradation of $Q=\sum_{i\in[m]}q_{i}\mathrm{B}(\sigma_{i})\in \mathbb{B}_{m}^{\ast}$, where $m>n\geq2$ and $0\leq \sigma_{1}<\cdots <\sigma_{m}\leq1/2$. There exsist $1\leq i_{2}<\cdots <i_{n}\leq m$ and $0\leq s_{j}\leq q_{i_{j}}, j=2, 3, \ldots, n$, such that
\begin{equation}\label{28}
W\preccurlyeq_{\text{P}} D_{Q}(i_{2},\ldots, i_{n}; s_{2},\ldots, s_{n})=\sum_{j\in[n]}p_{j}\mathrm{B}(\varepsilon_{j})\preccurlyeq_{\text{P}} Q,
\end{equation}
where $0\leq \varepsilon_{1}<\cdots <\varepsilon_{n}\leq1/2$. Furthermore, if $i_{j+1}=i_{j}+1$ for some $j\in[n]$, then we have\\

1. $s_{j+1}=0$ and $p_{j}\mathrm{B}(\varepsilon_{j})=q_{i_{j+1}}\mathrm{B}(\sigma_{i_{j+1}})$ when $s_{j}=0$,\\

2. $s_{j}=q_{i_{j}}$ and $p_{j}\mathrm{B}(\varepsilon_{j})=q_{i_{j}}\mathrm{B}(\sigma_{i_{j}})$ when $s_{j+1}=q_{i_{j+1}}$.\\

\noindent
Therefore, we have $p_{j}\geq q_{i}$ if $\varepsilon_{j}=\sigma_{i}$. In particular, we have $p_{1}\geq q_{1}$ and $p_{n}\geq q_{m}$.

\end{theorem}

\begin{proof}
According to Theorems \ref{the01} and \ref{the02}, any $2n$-degradation of $Q$ is a degradation of some $2n$-P-degradation of $Q$. Therefore, without loss of generality, we assume further that $W$ is a P-degradation of $Q$. Let $K=K(W)=(k_{i,j})_{m\times n}$ be the 1-matrix given in Theorem \ref{the02} for the channel $W$.

If there are integers $j, j'$ such that
\begin{equation*}
\left(\sum_{i\in[m]}k_{i,j'}\right)\sum_{i\in[m]}k_{i,j}\sigma_{i}=\left(\sum_{i\in[m]}k_{i,j}\right)\sum_{i\in[m]}k_{i,j'}\sigma_{i},
\end{equation*}
we replace $K$ with a matrix that has one less column, which is obtained by substituting the $j$-th column and the $j'$-th column with their sum vector. Obviously, this operation is equivalent to replacing the channel $W$ with one of its equivalent channels.

If there are at least three nonzero entries in the $i$-th row of $K$ for some $i\in[m]$, then according to the first two conclusions of Lemma \ref{lem06} one can get a channel $W_{0}$ which is a P-degradation of $Q$ and an upgradation of $W$, and the corresponding matrix $K(W_{0})$ has one more zero entry.

Hence, by recursively applying these two kinds of operations on the resulting channels one can get a channel $W_{1}$, which is P-degradation of $Q$ and upgradation of $W$, such that the corresponding matrix $K(W_{1})=(k_{i,j}^{(1)})_{m\times t}$ has at most two nonzero entries in each row, the numbers $\varepsilon_{1},\ldots,\varepsilon_{t}$ determined by
\begin{equation*}
\varepsilon_{j}\sum_{i\in [m]}k_{i,j}^{(1)}=\sum_{i\in [m]}k_{i,j}^{(1)}\sigma_{i}
\end{equation*}
satisfy $0\leq \varepsilon_{1}<\cdots<\varepsilon_{t}\leq1/2$ and
\begin{align*}
	\varepsilon_{j-1}<\sigma_{i}<\varepsilon_{j+1} & \text{ if }k_{i,j}^{(1)}\neq0,\\
     k_{i,j-1}^{(1)}+k_{i,j}^{(1)}=q_{i} & \text{ if }\varepsilon_{j-1}<\sigma_{i}<\varepsilon_{j},\\
	 k_{i,j}^{(1)}=q_{i} & \text{ if }\varepsilon_{j}=\sigma_{i},
\end{align*}
where $t$ is a positive integer not greater than $n$.

Furthermore, according to the third conclusion of Lemma \ref{lem06} one can require further the channel $W_{1}$ satisfying that, for each $j$ with $2\leq j\leq t$, there is a $\delta\in(\varepsilon_{j-1}, \varepsilon_{j})$ such that $k_{i,j-1}^{(1)}=q_{i}$ for each $i\in[m]$ with $\varepsilon_{j-1}<\sigma_{i}<\delta$, and $k_{i',j}^{(1)}=q_{i'}$ for each $i'\in[m]$ with $\delta<\sigma_{i'}<\varepsilon_{j}$.

Clearly, the desired channel $D_{Q}(i_{2},\ldots, i_{n}; s_{2},\ldots, s_{n})$ can be obtained from $W_{1}$ simply, by splitting some groups (\ref{27}) of particles into subgroups properly if $t < n$.
\end{proof}

Any channel $D_{Q}(i_{2},\ldots, i_{n}; s_{2},\ldots, s_{n})$ that satisfies the conditions of Theorem \ref{the03} is referred to as a \emph{$2n$-$P^{\ast}$-degradation} of $Q$. Since every $2n$-P-degradation of $Q$ is a degradation of some $2n$-$\mathrm{P}^{\ast}$-degradation of $Q$, our discussion is restricted to the $\mathrm{P}^{\ast}$-degradations of $Q$. The following theorem demonstrates that, for any $2n$-$\mathrm{P}^{\ast}$-degradation $W$ of a given $Q\in \mathbb{B}_{m}^{\ast}$, at least three splitting patterns must be adjusted to derive another $2n$-$\mathrm{P}^{\ast}$-degradation $W'$ of $Q$ such that $W\preccurlyeq W'$.
\begin{theorem}\label{the04}
For $2\leq n<m$, from a $2n$-$P^{\ast}$-degradation $W$ of given $Q\in \mathbb{B}_{m}^{\ast}$, it is impossible to obtain a new $2n$-$P^{\ast}$-degradation $W'$ of $Q$ with $W\preccurlyeq W'$ by adjusting at most two splitting patterns of $W$.
\end{theorem}
\begin{proof}
The proof is given in Appendix C.
\end{proof}
\section{Degradations with the Highest Symmetric Capacity}
\label{sec06}
For any given symmetric BIDMC, it has been shown in the last section that every P-degradation is a degradation of some $\mathrm{P}^{\ast}$-degradation. In this section we consider to determine the degradations with the highest symmetric
capacity from the $\mathrm{P}^{\ast}$-degradations.

For $\sigma\in[0,1/2]$, $\varepsilon\in(0,1/2)$ and $p\in(0,1)$, let
\begin{equation*}
b_{\sigma,\varepsilon}=\min\{1,\varepsilon/\sigma,(1-2\varepsilon)/(1-2\sigma)\}.
\end{equation*}
Then, for any $x\in(-pb_{\sigma,\varepsilon},1-p)$, we have $0<p+x<1$ and
\begin{equation*}
e_{\sigma,\varepsilon,p}(x)=(p\varepsilon+x\sigma)/(p+x)\in(0,1/2).
\end{equation*}
For $x\in(-pb_{\sigma,\varepsilon},1-p)$, let
\begin{equation}\label{29}
f_{\sigma,\varepsilon,p}(x)=(p+x)f_{c}(e_{\sigma,\varepsilon,p}(x)),
\end{equation}
where $f_{c}(p)=1+p\log p+\overline{p}\log \overline{p}$. Then, from $f'_{c}(p)=\log (p/\overline{p})$ we see
\begin{align}
f'_{\sigma,\varepsilon,p}(x)
=&f_{c}(e_{\sigma,\varepsilon,p}(x))+\log\frac{e_{\sigma,\varepsilon,p}(x)}{\overline{e_{\sigma,\varepsilon,p}(x)}}\frac{\sigma(p+x)-(p\varepsilon+x\sigma)}{p+x}\nonumber\\
=&1+\frac{p\varepsilon+x\sigma}{p+x}\log \frac{p\varepsilon+x\sigma}{p+x}+\frac{p\overline{\varepsilon}+x\overline{\sigma}}{p+x}\log \frac{p\overline{\varepsilon}+x\overline{\sigma}}{p+x}\nonumber\\
&+\frac{p(\sigma-\varepsilon)}{p+x}\log \frac{p\varepsilon+x\sigma}{p\overline{\varepsilon}+x\overline{\sigma}}\nonumber\\
=&1+\sigma\log (p\varepsilon+x\sigma)+\overline{\sigma}\log (p\overline{\varepsilon}+x\overline{\sigma})-\log (p+x),\label{30}\\
f''_{\sigma,\varepsilon,p}(x)
=&\frac{\sigma^{2}}{(p\varepsilon+x\sigma)\ln 2}+\frac{\overline{\sigma}^{2}}{(p\overline{\varepsilon}+x\overline{\sigma})\ln 2}-\frac{1}{(p+x)\ln2}\nonumber\\
=&\frac{p^{2}(\sigma-\varepsilon)^{2}}{(p+x)(p\varepsilon+x\sigma)(p\overline{\varepsilon}+x\overline{\sigma})\ln 2}.\label{31}
\end{align}
For $0<\varepsilon_{1}<\sigma<\varepsilon_{2}<1/2$ and positive numbers $p_{1}, p_{2}$ with $p_{1}+p_{2}\leq1$, let
\begin{equation}\label{32}
F_{\sigma,\varepsilon_{1},p_{1},\varepsilon_{2},p_{2}}(x)=f_{\sigma,\varepsilon_{1},p_{1}}(-x)+f_{\sigma,\varepsilon_{2},p_{2}}(x).
\end{equation}
Then, from (\ref{30}) and (\ref{32}) we see
\begin{equation}\label{33}
F'_{\sigma,\varepsilon_{1},p_{1},\varepsilon_{2},p_{2}}(0)=\log \frac{\overline{\varepsilon_{2}}}{\overline{\varepsilon_{1}}}+\sigma\log \frac{\overline{\varepsilon_{1}}\varepsilon_{2}}{\overline{\varepsilon_{2}}\varepsilon_{1}}=(\sigma-\phi(\varepsilon_{1},\varepsilon_{2}))\log \frac{\overline{\varepsilon_{1}}\varepsilon_{2}}{\overline{\varepsilon_{2}}\varepsilon_{1}} ,
\end{equation}
where
\begin{equation}\label{34}
\phi(\varepsilon_{1},\varepsilon_{2})=\frac{\ln (\overline{\varepsilon_{1}}/\overline{\varepsilon_{2}})}{\ln ((\overline{\varepsilon_{1}}\varepsilon_{2})/(\overline{\varepsilon_{2}}\varepsilon_{1}))}.
\end{equation}

\begin{lemma}\label{lem07}
For $0<\varepsilon_{1}<\varepsilon_{2}<1/2$, we have
\begin{equation}\label{35}
\varepsilon_{1}<\phi(\varepsilon_{1},\varepsilon_{2})<\varepsilon_{2}.
\end{equation}
Furthermore, for positive numbers $\sigma, p_{1}, p_{2}$ with $\varepsilon_{1}\leq\sigma\leq\varepsilon_{2}$ and $p_{1}+p_{2}\leq1$,\\

1. For any $x$ in $(-p_{2}b_{\sigma,\varepsilon_{2}}, p_{1}b_{\sigma,\varepsilon_{1}})=(-p_{2}(1-2\varepsilon_{2})/(1-2\sigma), p_{1}\varepsilon_{1}/\sigma)$,
\begin{equation}\label{36}
0<e_{\sigma,\varepsilon_{1},p_{1}}(-x)<e_{\sigma,\varepsilon_{2},p_{2}}(x)<1/2.
\end{equation}

2. $F_{\sigma,\varepsilon_{1},p_{1},\varepsilon_{2},p_{2}}(x)$ is decreasing on $(-p_{2}b_{\sigma,\varepsilon_{2}},0]$ if $\sigma\leq\phi(\varepsilon_{1},\varepsilon_{2})$.\\

3. $F_{\sigma,\varepsilon_{1},p_{1},\varepsilon_{2},p_{2}}(x)$ is increasing on $[0, p_{1}b_{\sigma,\varepsilon_{1}})$ if $\sigma\geq\phi(\varepsilon_{1},\varepsilon_{2})$.
\end{lemma}
\begin{proof}
The proof is given in Appendix D.
\end{proof}

The \emph{capacity-loss rate} (CLR) of $\mathrm{P}^{\ast}$-degradation $W$ of symmetric BIDMC $Q$ is the ratio $(I(Q)-I(W))/I(Q)$. If $W$ is a $2n$-$\mathrm{P}^{\ast}$-degradation of $Q$, we say another $2n$-$\mathrm{P}^{\ast}$-degradation $W'$ of $Q$ is \emph{better} than $W$ if $W'$ has smaller CLR. A $2n$-$\mathrm{P}^{\ast}$-degradation of $Q$ is said to be \emph{C-optimal} if its CLR achieves the minimal.

\begin{lemma}\label{lem08}
For $Q=\sum_{i\in[m]}q_{i}\mathrm{B}(\sigma_{i})\in\mathbb{B}_{m}^{\ast}$ with $0\leq \sigma_{1}<\cdots<\sigma_{m}\leq1/2$ and $n$ with $2\leq n<m$, let $D=D_{Q}(i_{2},\ldots, i_{n}; s_{2},\ldots, s_{n})=\sum_{j\in[n]}p_{j}\mathrm{B}(\varepsilon_{j})$ be an
arbitrary $2n$-$P^{\ast}$-degradation of $Q$, where $0\leq \varepsilon_{1}<\cdots<\varepsilon_{n}\leq1/2$. For any integer $k$ with $2\leq k\leq n$,

\begin{enumerate}
    \item If $\varepsilon_{k-1} = 0$, then we have $k = 2$, $\sigma_1 = 0$, $i_2 = 2$, $s_2 = q_2$ and $p_1 = q_1$.
    \item If $\varepsilon_k = 1/2$, then we have $k = n$, $\sigma_m = 1/2$, $i_n = m - 1$, $s_n = 0$ and $p_n = q_m$.
    \item If $0 < \varepsilon_{k-1} < \varepsilon_k < 1/2$, then a better $P^*$-degradation of $Q$ may be obtained from $D$ by adjusting the splitting pattern $(i_k,s_k)$ in one of the following two ways
    \begin{enumerate}
        \item[a)] Replacing $s_k$ with $0$ if $s_k > 0$ and $\sigma_{i_k} \leq \phi(\varepsilon_{k-1},\varepsilon_k)$, or
        \item[b)] Replacing $s_k$ with $q_{i_k}$ if $s_k < q_{i_k}$ and $\sigma_{i_k} \geq \phi(\varepsilon_{k-1},\varepsilon_k)$.
    \end{enumerate}
\end{enumerate}
\end{lemma}

\begin{proof}
The first two conclusions follow from Theorem \ref{the03} directly.

To prove the third conclusion of this lemma for the case $a)$, we assume $0<\varepsilon_{k-1}<\varepsilon_{k}<1/2, s_{k}>0$ and $\sigma_{i_{k}}\leq\phi(\varepsilon_{k-1},\varepsilon_{k})$. Then, by using (\ref{35}) we see easily
\begin{equation*}
0<\varepsilon_{k-1}\leq\sigma_{i_{k}}\leq\phi(\varepsilon_{k-1},\varepsilon_{k})<\varepsilon_{k}=(s_{k}\sigma_{i_{k}}+(p_{k}-s_{k})\varepsilon)/p_{k}
\end{equation*}
for some $\varepsilon\in(\varepsilon_{k},1/2]$, and thus we have
\begin{equation*}
0<s_{k}=p_{k}(\varepsilon-\varepsilon_{k})/(\varepsilon-\sigma_{i_{k}})\leq p_{k}(1-2\varepsilon_{k})/(1-2\sigma_{i_{k}})=p_{k}b_{\sigma_{i_{k}},\varepsilon_{k}}.
\end{equation*}
Therefore, for the channel $D_{2}$ obtained from $D$ by replacing the splitting pattern $(i_{k},s_{k})$ with $(i_{k},0)$, according to Lemma \ref{lem07} and $\sigma_{i_{k}}\leq \phi(\varepsilon_{k-1},\varepsilon_{k})$ we see
\begin{equation*}
I(D_{2})-I(D)=F_{\sigma_{i_{k}},\varepsilon_{k-1},p_{k-1},\varepsilon_{k},p_{k}}(-s_{k})-F_{\sigma_{i_{k}},\varepsilon_{k-1},p_{k-1},\varepsilon_{k},p_{k}}(0)>0.
\end{equation*}
Hence, the third conclusion of this lemma is true for the case $a)$.

To prove the third conclusion of this lemma for the case $b)$, we assume $0<\varepsilon_{k-1}<\varepsilon_{k}<1/2, s_{k}<q_{i_{k}}$ and $\sigma_{i_{k}}\geq\phi(\varepsilon_{k-1},\varepsilon_{k})$. Then, by using (\ref{35}) we see easily
\begin{equation*}
\varepsilon_{k}\geq\sigma_{i_{k}}\geq\phi(\varepsilon_{k-1},\varepsilon_{k})>\varepsilon_{k-1}=((q_{i_{k}}-s_{k})\sigma_{i_{k}}+(p_{k-1}+s_{k}-q_{i_{k}})\varepsilon')/p_{k-1}>0
\end{equation*}
for some $\varepsilon'\in[0,\sigma_{i_{k}})$, and thus we have $0<q_{i_{k}}-s_{k}<p_{k-1}$. Therefore, for the channel $D'_{2}$ obtained from $D$ by replacing the splitting pattern $(i_{k},s_{k})$ with $(i_{k},q_{i_{k}})$, from Lemma \ref{lem07}, $\sigma_{i_{k}}\geq \phi(\varepsilon_{k-1},\varepsilon_{k})$ and
\begin{equation*}
0<q_{i_{k}}-s_{k}=(p_{k-1}\varepsilon_{k-1}-(p_{k-1}+s_{k}-q_{i_{k}})\varepsilon')/\sigma_{i_{k}}\leq p_{k-1}\varepsilon_{k-1}/\sigma_{i_{k}}=b_{\sigma_{i_{k}},\varepsilon_{k-1}}p_{k-1}
\end{equation*}
we see
\begin{equation*}
I(D'_{2})-I(D)=F_{\sigma_{i_{k}},\varepsilon_{k-1},p_{k-1},\varepsilon_{k},p_{k}}(q_{i_{k}}-s_{k})-F_{\sigma_{i_{k}},\varepsilon_{k-1},p_{k-1},\varepsilon_{k},p_{k}}(0)>0.
\end{equation*}
Hence, the third conclusion of this lemma is true for the case $b)$.
\end{proof}

\begin{rmk}\label{rem02}
The third conclusion of Lemma \ref{lem08} implies that, for the case $\sigma_{i_{k}}=\phi(\varepsilon_{k-1},\varepsilon_{k})$, one can obtain two better $P^{\ast}$-degradations from $D$ by replacing the splitting pattern $(i_{k},s_{k})$ with $(i_{k},0)$ and $(i_{k},q_{i_{k}})$ respectively. This also indicates that the CLR of $D$ reaches a local maximum.
\end{rmk}

For any symmetric BIDMC $Q=\sum_{i\in[m]}q_{i}\mathrm{B}(\sigma_{i})$ with $0\leq \sigma_{1}<\cdots<\sigma_{m}\leq1/2$, its $2n$-$\mathrm{P}^{\ast}$-degradation $D_{Q}(i_{2},\ldots,i_{n};s_{2},\ldots,s_{n})$ is referred to as a \emph{$2n$-$P^{+}$-degradation} if and only if $s_{j}\in\{0,q_{i_{j}}\}$ for all $j=2,3, \ldots, n$. Clearly,
for any $2n$-P$^{+}$-degradation $W=\sum_{j\in[n]}p_{j}\mathrm{B}(\varepsilon_{j})$ of $Q$, there exist integers $1<k_{1}<\cdots<k_{n-1}\leq m$ such that for each $j\in [n]$ the particle $p_{j}\mathrm{B}(\varepsilon_{j})$ is
compounded from the particles in $\{q_{i}\mathrm{B}(\sigma_{i}): k_{j-1}\leq i<k_{j}\}$, where $k_{0}=1$ and $k_{n}=m+1$. Such a $2n$-P$^{+}$-degradation $W$ is denoted by $D_{Q}^{+}(k_{1},\ldots,k_{n-1})$. From Lemma \ref{lem08}, the following corollary can be directly deduced.

\begin{corollary}\label{cor03}
Suppose $2\leq n<m$ and $Q=\sum_{i\in[m]}q_{i}\mathrm{B}(\sigma_{i})\in \mathbb{B}_{m}^{\ast}$, where $0\leq \sigma_{1}<\cdots<\sigma_{m}\leq1/2$.

\begin{enumerate}
    \item Any $C$-optimal $2n$-degradation of $Q$ must be a $2n$-$P^+$-degradation of $Q$.

    \item From a $2n$-$P^+$-degradation $W = D_Q^+(k_1, \dots, k_{n-1}) = \sum_{j \in [n]} p_j B(\varepsilon_j)$ of $Q$, where $k_0 = 1 < k_1 < \dots < k_n = m+1$ and $0 \leq \varepsilon_1 < \dots < \varepsilon_n \leq 1/2$, a better $2n$-$P^+$-degradation of $Q$ may be obtained by
    \begin{enumerate}
        \item[(a)] adding $1$ to $k_j$ if $\phi(\varepsilon_j, \varepsilon_{j+1}) \geq \sigma_{k_j}$, or
        \item[(b)] subtracting $1$ from $k_j$ if $\phi(\varepsilon_j, \varepsilon_{j+1}) \leq \sigma_{k_j-1}$ and $\varepsilon_j + \sigma_{k_j-1} > 0$.
    \end{enumerate}
\end{enumerate}
\end{corollary}

\begin{proof}
It is obvious that, for any $1< k_{1}<\cdots<k_{n-1}\leq m$,
\begin{align*}
D_{Q}^{+}(k_{1},\ldots, k_{n-1})
\cong&D_{Q}(k_{1}-1,\ldots, k_{n-1}-1;0,\ldots, 0)\\
\cong&D_{Q}(k_{1},\ldots, k_{n-1};q_{k_{1}},\ldots, q_{k_{n-1}}).
\end{align*}
Hence, this corollary follows from Lemma \ref{lem08} simply.
\end{proof}

Therefore, we have the following theorem which gives some necessary conditions for the optimal $2n$-degradations of a symmetric BIDMC.

\begin{theorem}\label{the05}
Let $W$ be a C-optimal $2n$-degradation of $Q=\sum_{i\in[m]}q_{i}\mathrm{B}(\sigma_{i})\in \mathbb{B}_{m}^{\ast}$, where $2\leq n<m$ and $0\leq \sigma_{1}<\cdots<\sigma_{m}\leq1/2$. Then, there are integers $k_{0}=1<k_{1}<\cdots<k_{n}=m+1$ such that $W\cong D_{Q}^{+}(k_{1},\ldots,k_{n-1})=\sum_{j\in[n]}p_{j}\mathrm{B}(\varepsilon_{j})$, for each $j\in [n]$ the particle $p_{j}\mathrm{B}(\varepsilon_{j})$ is compounded from the particles in $\{q_{i}\mathrm{B}(\sigma_{i}): k_{j-1}\leq i<k_{j}\}$, and for each $j\in[n-1]$ with $\varepsilon_{j}>0$
\begin{equation}\label{37}
\sigma_{k_{j}-1}<\phi(\varepsilon_{j}, \varepsilon_{j+1})<\sigma_{k_{j}}.
\end{equation}
\end{theorem}

\begin{proof}
This theorem follows simply from Corollary \ref{cor03}.
\end{proof}

The degradations $W$ of $Q\in \mathbb{B}_{m}$ fulfilling the necessary conditions shown in Theorem \ref{the05} are referred to as \emph{$2n$-C-degradations} of $Q$. Clearly, the number of $2n$-C-degradations of $Q$ is not greater than the number of $2n$-P$^{+}$-degradations of $Q$, which is given by $\binom{m-1}{n-1}$. Notice that the number of
$2n$-C-degradations of a given symmetric BIDMC exceeds 1 in general. For example, given $0<\sigma_{1}<\sigma_{2}<\sigma_{3}<1/2$, according to (\ref{35}) there exists a small number $\delta>0$ such that, for any positive numbers $q_{1}, q_{3}$ in the interval $(0,\delta)$, when we define $\varepsilon_{1}=(q_{2}\sigma_{2}+q_{3}\sigma_{3})/(q_{2}+q_{3})$ and $\varepsilon_{2}=(q_{1}\sigma_{1}+q_{2}\sigma_{2})/(q_{1}+q_{2})$ with $q_{2}=1-q_{1}-q_{3}$, the following inequalities hold:
\begin{equation*}
\sigma_{1}<\phi(\sigma_{1}, \varepsilon_{1})<\sigma_{2}<\phi(\varepsilon_{2},\sigma_{3})<\sigma_{3}.
\end{equation*}
Consequently, $W_{1}=q_{1}\mathrm{B}(\sigma_{1})+(q_{2}+q_{3})\mathrm{B}(\varepsilon_{1})$ and $W_{2}=(q_{1}+q_{2})\mathrm{B}(\varepsilon_{2})+q_{3}\mathrm{B}(\sigma_{3})$ are 4-C-degradations of $Q=\sum_{i=1}^{3}q_{i}\mathrm{B}(\sigma_{i})$.
\section{Some Simulation Results}
\label{sec07}
The following algorithm outputs all the $2n$-C-degradations of a given symmetric BIDMC in $\mathbb{B}_{m}$.\\

\noindent
\textbf{Algorithm 1}: Input $2\leq n<m$, $0\leq\sigma_{1}<\cdots<\sigma_{m}\leq 1/2$ and positive numbers $q_{1},\ldots,q_{m}$ with $\sum_{i\in[m]}q_{i}=1$.

\vspace{1em}

S0. Let $j_{0}=1$, $\varepsilon_{1}=\sigma_{1}$, $p_{1}=q_{1}$, $i_{1}=1$, $j_{1}=2$, $r_{1}=0$, $\delta_{1}=0$, and $k=1$.

\vspace{0.5em}
\hspace{1.4em} If $\sigma_1 > 0$, goto S1.

\vspace{0.5em}
\hspace{1.4em} Let $\varepsilon_2 = \sigma_2$, $p_2 = q_2$, $i_2 = 2$, $j_2 = 3$, $r_2 = 0$, $\delta_2 = 0$ and $k = 2$.

\vspace{1em}

S1. If $k = 0$, goto S5.

\vspace{0.5em}
\hspace{1.4em} If $j_k > m - n + k + 1$, let $k \leftarrow -1$, $j_k \leftarrow +1$ and goto S1.

\vspace{0.5em}
\hspace{1.4em} Let $r_k \leftarrow +q_{j_k}$, $\delta_k \leftarrow +q_{j_k}(\sigma_{j_k} - \delta_k)/r_k$ and $\phi_k = \phi(\varepsilon_k, \delta_k)$.

\vspace{0.5em}
\hspace{1.4em} If $\phi_k \leq \sigma_{i_k}$, let $j_k \leftarrow +1$ and goto S1.

\vspace{0.5em}
\hspace{1.4em} If $\phi_k \geq \sigma_{i_{k}+1}$, let $k \leftarrow -1$, $j_k \leftarrow +1$ and goto S1.

\vspace{0.5em}
\hspace{1.4em} Let $k \leftarrow +1$, $\varepsilon_k = \delta_{k-1}$, $p_k = r_{k-1}$, $i_k = j_{k-1}$, $j_k = i_k + 1$, $r_k = 0$ and

\vspace{0.5em}
\hspace{1.4em} $\delta_k = 0$.

\vspace{0.5em}
\hspace{1.4em} If $k < n-1$, goto S1.

\vspace{0.5em}
\hspace{1.4em} Let $j = j_{n-1}$, $r = 0$ and $\delta = 0$.

\vspace{1em}

S2. Let $r \leftarrow +q_j$ and $\delta \leftarrow +q_j(\sigma_j - \delta)/r$.

\vspace{0.5em}
\hspace{1.4em} If $j < m$, let $j \leftarrow +1$ and goto S2.

\vspace{0.5em}
\hspace{1.4em} If $\sigma_{i_{n-1}} < \phi(\varepsilon_{n-1}, \delta) < \sigma_{i_{n-1}+1}$, goto S4.

\hspace{1em}

S3. Let $k = n-2$, $j_k \leftarrow +1$ and goto S1.

\hspace{1em}

S4. Let $i_n = m$, $\varepsilon_n = \delta$ and $p_n = r$.

\vspace{0.5em}
\hspace{1.4em} Output $(i_1, \varepsilon_1, p_1), \dots, (i_n, \varepsilon_n, p_n)$ and goto S3.

\hspace{1em}

S5. If $j_0 > m - n + 1$, stop.

\vspace{0.5em}
\hspace{1.4em} Let $j = j_0$, $p_1 \leftarrow +q_j$, $\varepsilon_1 \leftarrow +q_j(\sigma_j - \varepsilon_1)/p_1$, $i_1 = j$, $j_1 = j + 1$, $r_1 = 0$,

\vspace{0.5em}
\hspace{1.4em} $\delta_1 = 0$, $k = 1$ and goto S1.\\

For $m = 8$ and $n = 4$, the number of $\mathrm{P}^+$-degradations of any channel in $\mathbb{B}_8$ is $\binom{8-1}{4-1} = 35$, the mean number of C-degradations is about $6.749$ and the mean CLR of C-degradations is about $0.0048$. The following tabular shows the mean CLR of the optimal $(2m,2n)$-degradations of channels in $\mathbb{B}_n$ for $(n,m)$ with $n = 4,5,6,7,8,9,10$ and $m = 16,32,64,128$.

\[
\begin{array}{|c|c|c|c|c|c|c|c|}
\hline
m\backslash n & 4      & 5      & 6      & 7      & 8      & 9      & 10     \\
\hline
128           & 0.0232 & 0.0145 & 0.0097 & 0.0069 & 0.0051 & 0.0040 & 0.0030 \\
\hline
64            & 0.0223 & 0.0137 & 0.0091 & 0.0064 & 0.0047 & 0.0035 & 0.0027 \\
\hline
32            & 0.0182 & 0.0105 & 0.0065 & 0.0043 & 0.0029 & 0.0021 & 0.0015 \\
\hline
16            & 0.0120 & 0.0061 & 0.0032 & 0.0018 & 0.0010 & 0.0006 & 0.0003 \\
\hline
\end{array}
\]

Now we consider the application of our method to Arikan transformations.

Let $(u_0, u_1)$ be a random vector uniformly distributed over $\mathcal{X}^2$. Assume that $u_0 + u_1$ is transmitted over $W_0: x_0 \in \mathcal{X} \mapsto y_0 \in \mathcal{Y}_0$, whereas $u_1$ is transmitted over $W_1: x_1 \in \mathcal{X} \mapsto y_1 \in \mathcal{Y}_1$, respectively, where the BIDMCs $W_0$ and $W_1$ are independent copies of $W$. The synthetic channel $A_0(W)$: $u_0 \mapsto (y_0, y_1)$ with transition probabilities
\[
\Pr(y_0, y_1|u_0) = \frac{1}{2} \sum_{u_1 \in \mathcal{X}} \Pr(y_0|x_0 = u_0 + u_1) \Pr(y_1|x_1 = u_1)
\]
and the synthetic channel $A_1(W)$: $u_1 \mapsto (y_0, y_1, u_0)$ with transition probabilities
\[
\Pr(y_0, y_1, u_0|u_1) = \frac{1}{2} \Pr(y_0|x_0 = u_0 + u_1) \Pr(y_1|x_1 = u_1)
\]
are called \textit{Arikan Transformations} of $W$, which were used iteratively to construct polar codes in \cite{Arikan09}. It is well-known that, for $W' \preccurlyeq W$, we have
\begin{align}\label{38}
A_0(W') \preccurlyeq A_0(W),\ A_1(W') \preccurlyeq A_1(W).
\end{align}

It was shown in \cite{JCT24} that, for $W \cong \sum_{j \in [n]} p_j \mathrm{B}(\varepsilon_j) \in \mathbb{B}_n$, we have
\begin{align}\label{39}
A_0(W) \cong \sum_{i,j \in [n]} p_i p_j \mathrm{B}(\varepsilon_i \star \varepsilon_j),
\end{align}
\begin{align}\label{40}
A_1(W) \cong \sum_{i,j \in [n]} p_i p_j \big( (\overline{\varepsilon_i} \star \varepsilon_j) \mathrm{B}(\varepsilon_i \diamond \varepsilon_j) + (\varepsilon_i \star \varepsilon_j) \mathrm{B}(\overline{\varepsilon_i} \diamond \varepsilon_j) \big),
\end{align}
where $a \star b = \overline{a}b + a\overline{b}$ and
\[
a \diamond b =
\begin{cases}
ab/(\overline{a} \star b), & \text{if } a,b \in (0,1), \\
0, & \text{if } \{a,b\} \cap \{0,1\} \neq \emptyset.
\end{cases}
\]

As pointed in \cite{JCT24}, from (\ref{39}) and (\ref{40}) we see easily that Arikan transformations of any channel in $\mathbb{B}_n$ can be expressed as RSCs of at most $n^2 + 1$ BSCs. For $n = 4,5,6,7,8,9,10$ and $m = n^2 + 1$, the mean CLRs of some $(2m,2n)$-degradations are exhibited in the following tabular, where Opt-deg, TV-deg and TV$^{\ast}$-deg denote the optimal degradation, the degradation given in \cite{Vardy13} and its optimization obtained according to the conclusion 2 of Corollary \ref{cor03}, respectively.

\[
\begin{array}{|c|c|c|c|c|c|c|c|}
\hline
\backslash n & 4      & 5      & 6      & 7      & 8      & 9      & 10     \\
\hline
\text{Opt-deg}  & 0.0127 & 0.0093 & 0.0072 & 0.0057 & 0.0047 & 0.0039 & 0.0031 \\
\hline
\text{TV-deg}   & 0.0134 & 0.0099 & 0.0077 & 0.0062 & 0.0051 & 0.0043 & 0.0034 \\
\hline
\text{TV$^{\ast}$-deg}  & 0.0130 & 0.0095 & 0.0074 & 0.0059 & 0.0049 & 0.0041 & 0.0032 \\
\hline
\end{array}
\]

For $n = 4,5,6,7,8,9,10$ and $m = n^2 + 1$, the mean number and the mean CLR of C-$(2m,2n)$-degradations are exhibited in the following tabular.

\[
\begin{array}{|c|c|c|c|c|c|c|c|}
\hline
\backslash n & 4       & 5       & 6      & 7       & 8       & 9       & 10      \\
\hline
\text{Number of C-deg} & 14.526  & 47.403  & 161.59 & 567.19  & 2033.4  & 7374.1  & 26332   \\
\hline
\text{CLR of C-deg}   & 0.0288  & 0.0208  & 0.0155 & 0.0119  & 0.0093  & 0.0075  & 0.0064  \\
\hline
\end{array}
\]

For $a \in \{0,1\}$, binary sequence $\alpha \in \{0,1\}^*$ and symmetric BIDMC $W \in \mathbb{B}_n$, let $A_{\alpha a}(W)$ and $D_{n,\alpha a}^{\text{opt}}(W)$ denote $A_a(A_\alpha(W))$ and the optimal $2n$-degradation of $A_a(D_{n,\alpha}^{\text{opt}}(W))$, respectively, where $A_\eta(W) = D_{n,\eta}^{\text{opt}}(W) = W$ for the empty sequence $\eta$. In the following tabular, for $\alpha \in \{0,1\}^*$, $\text{CLR}_{n,\alpha}$ denotes the mean value of $(I(A_\alpha(W)) - I(D_{n,\alpha}^{\text{opt}}(W)))/I(A_\alpha(W))$ for $W \in \mathbb{B}_n$.

\[
\begin{array}{|c|c|c|c|c|c|c|c|}
\hline
\backslash n & 4       & 5       & 6       & 7       & 8       & 9       & 10      \\
\hline
\text{CLR}_{n,0}      & 0.0065  & 0.0048  & 0.0038  & 0.0031  & 0.0026  & 0.0023  & 0.0020  \\
\hline
\text{CLR}_{n,1}      & 0.0094  & 0.0061  & 0.0044  & 0.0034  & 0.0027  & 0.0022  & 0.0011  \\
\hline
\text{CLR}_{n,00}     & 0.0318  & 0.0241  & 0.0192  & 0.0155  & 0.0136  & 0.0112  & 0.0101  \\
\hline
\text{CLR}_{n,01}     & 0.0253  & 0.0183  & 0.0138  & 0.0108  & 0.0090  & 0.0074  & 0.0063  \\
\hline
\text{CLR}_{n,10}     & 0.0249  & 0.0177  & 0.0133  & 0.0103  & 0.0085  & 0.0069  & 0.0057  \\
\hline
\text{CLR}_{n,11}     & 0.0169  & 0.0109  & 0.0077  & 0.0058  & 0.0046  & 0.0036  & 0.0030  \\
\hline
\text{CLR}_{n,000}    & 0.1026  & 0.0775  & 0.0637  & 0.0524  & 0.0473  & 0.0365  & 0.0329  \\
\hline
\text{CLR}_{n,001}    & 0.0563  & 0.0407  & 0.0329  & 0.0268  & 0.0234  & 0.0181  & 0.0163  \\
\hline
\text{CLR}_{n,010}    & 0.0666  & 0.0482  & 0.0383  & 0.0300  & 0.0260  & 0.0201  & 0.0190  \\
\hline
\text{CLR}_{n,011}    & 0.0419  & 0.0291  & 0.0224  & 0.0172  & 0.0144  & 0.0111  & 0.0106  \\
\hline
\text{CLR}_{n,100}    & 0.0620  & 0.0439  & 0.0344  & 0.0262  & 0.0227  & 0.0174  & 0.0163  \\
\hline
\text{CLR}_{n,101}    & 0.0370  & 0.0252  & 0.0192  & 0.0145  & 0.0121  & 0.0093  & 0.0087  \\
\hline
\text{CLR}_{n,110}    & 0.0348  & 0.0225  & 0.0167  & 0.0122  & 0.0099  & 0.0075  & 0.0061  \\
\hline
\text{CLR}_{n,111}    & 0.0184  & 0.0112  & 0.0080  & 0.0057  & 0.0046  & 0.0032  & 0.0026  \\
\hline
\end{array}
\]

Next, we will integrate the algorithms in \cite{Yagi14} and \cite{Ozawa14}, incorporate the conditions of Theorem \ref{the05} into them, and present an algorithm for finding C-optimal $2n$-degradation. First, to utilize the dynamic programming-based quantizer design algorithm proposed in \cite{Yagi14}, we need to define a state value $S_{\varepsilon_{j}}(\sigma_{i})$ of the algorithm, which denotes the maximum partial symmetric capacity when $\sum_{i'\in[i]}q_{i'}\mathrm{B}(\sigma_{i'})$ is degraded to $\sum_{j'\in[j]}p_{j'}\mathrm{B}(\varepsilon_{j'})$. This value can be obtained via recursive computation:\\
For $j\leq i\leq j+m-n$, $1<j\leq n$,
\begin{align}\label{0041}
S_{\varepsilon_{j}}(\sigma_{i})=\max_{\substack{j-1 \leq a < i}}\{S_{\varepsilon_{j-1}}(\sigma_{a})+\iota(a+1,i)\},
\end{align}
where $\iota(a+1,i)=1+\varepsilon \log\varepsilon+\overline{\varepsilon}\log\overline{\varepsilon}, \varepsilon=(\sum_{i'=a+1}^{i}q_{i'}\sigma_{i'})/\sum_{i'=a+1}^{i}q_{i'}.$ In particular, for $j=1$, each $S_{\varepsilon_{1}}(\sigma_{i})$, $1\leq i\leq1+m-n$ is initialized by $S_{\varepsilon_{1}}(\sigma_{i})=\iota(1,i)$.

Subsequently, to efficiently solve for $S_{\varepsilon_{j}}(\sigma_{i})$ using the SMAWK algorithm in \cite{Wilber87}, we need to define an $r\times c$ totally monotone matrix $M$. The SMAWK algorithm can find the row maxima for all $r$ rows of matrix $M$ in $O(r+c)$ time. We directly adopt the definition of matrix $M$ from \cite{Ozawa14}, which is presented as follows:\\
For an $(m-n+1)$-order matrix $M_{j}$, $j\in\{2,3,\ldots,n\}$, its lower triangular part is defined as
\begin{align}\label{0042}
(M_{j})_{(a,b)}\stackrel{\text{def}}{=}S_{\varepsilon_{j-1}}(\sigma_{j-1+b})+\iota(j+b,j+a) ,
\end{align}
for $a\geq b$, $a,b\in\{0,1,\ldots,m-n\}$, where $(M_{j})_{(a,b)}$ denote the real value at the row $a$ and column $b$ of matrix $M_{j}$. Its upper triangular part is defined as
\begin{align}\label{0043}
(M_{j})_{(a,b)}\stackrel{\text{def}}{=}-2(b-a)+1,
\end{align}
for $a< b$, $a\in\{0,1,\ldots,m-n-1\}$, $b\in\{1,2,\ldots,m-n\}$.

Now, combining the conditions of Theorem \ref{the05}, we can present the following algorithm for finding C-optimal $2n$-degradation.\\

\noindent
\textbf{Algorithm 2}: Input $2\leq n<m$, $0\leq\sigma_{1}<\cdots<\sigma_{m}\leq 1/2$ and positive numbers $q_{1},\ldots,q_{m}$ with $\sum_{i\in[m]}q_{i}=1$.

\begin{enumerate}
    \item \textbf{Precompute partial mutual information}:

         Compute $\iota(i,i')$ for each $i\in\{1,2,\ldots,m\}$ and $i'\in\{i,\ldots,t\}$, where $t=\min\{i+m-n,m\}$.

    \item \textbf{Initialize}:

         $S_{\varepsilon_{1}}(\sigma_{i})=\iota(1,i)$ for each $i\in\{1,2,\ldots,1+m-n\}$.

    \item \textbf{Incorporate the decision conditions of Theorem \ref{the05}}:

          For each $k_{1}\in\{2,3,\ldots,m-n+2\}$, find all $k_{2}\in\{k_{1}+1,\ldots,m-n+3\}$ that satisfy the condition in (\ref{37}).

    \item \textbf{Process matrix $M_{j}$}:

          $\bullet$ For $j=2$, Step 3 specifies that only the element $(M_{j})_{(k_{2}-3,k_{1}-2)}$ corresponding to $(k_{1}, k_{2})$ needs computation; all other elements are omitted.

          $\bullet$ For $3\leq j\leq n-1$, if there exists an $i$-th row in $M_{j-1}$ with no maximum value (i.e., all elements in this row are skipped and not computed), then all elements in the $i$-th column of $M_{j}$ also do not need to be computed.

     \item \textbf{Compute and store decision}:

          $\bullet$ For each $j\in\{2,3,\ldots,n-1\}$, compute $S_{\varepsilon_{j}}(\sigma_{i})$ and store
          \begin{align*}
           d_j(i+1) \stackrel{\text{def}}{=} \underset{j-1 \leq a < i}{\arg\max} \left\{ S_{\varepsilon_{j-1}}(\sigma_a) + \iota(a+1, i) \right\},
           \end{align*}
          for $i\in\{j,j+1,\ldots,j+m-n\}$, using the SMAWK algorithm for finding the row maxima of totally monotone matrices defined by (\ref{0042}).

          $\bullet$ For $j=n$, compute $S_{\varepsilon_{n}}(\sigma_{m})$ according to (\ref{0041}) and store $d_n(m+1)$.

     \item \textbf{Traceback}:
          Let $k_{n}=m+1$, and for each $j\in\{n-1,n-2,\ldots,1\}$:
       \begin{align*}
           k_{j}=d_{j+1}(k_{j+1})+1.
       \end{align*}
\end{enumerate}

When performing Step 3, optimization strategies such as memory preallocation, pre-extraction of common variables, elimination of redundant computation, and replacement of nested loops with vectorized operations can be adopted to greatly reduce the constant factor overhead of actual operation, thereby significantly improving the execution efficiency and engineering practicability of the algorithm in large-scale input scenarios. Specifically, this preprocessing step completely eliminates nested loops through vectorized operations and only requires a single linear traversal, so its time complexity is $O(m)$. The core of Algorithm 2 lies in Step 4. From the previous calculations, it follows that the number of $2n$-C-degradations satisfying the conditions of Theorem \ref{the05} is already much smaller than the total number of $2n$-P$^{+}$-degradations. Therefore, executing this step allows us to skip the computation of most elements in the matrix $M_{j}$. On this basis, when the SMAWK algorithm is called to search for the maximum element in each row of matrix $M_{j}$, all uncomputed elements can be skipped directly, thereby greatly narrowing the search range. Its time complexity can be approximately expressed as $O(n')$, where $n'$ represents the scale of elements within the valid search range and satisfies $n'<m-n+1$. Especially when the value of $m-n+1$ is large, $n'<<m-n+1$, which significantly reduces the actual computational overhead of the algorithm. Therefore, the overall complexity of Algorithm 2 is $O(nn')$, which is superior to the complexity $O(n(m-n))$ of the algorithm in \cite{Ozawa14}.

As shown in Figure \ref{fig04}, for channel $Q=\sum_{i\in[n^{2}+1]}q_{i}\mathrm{B}(\sigma_{i})$, the speed of Algorithm 2 in finding the C-optimal $2n$-degradation is significantly superior to that of the algorithm proposed in \cite{Ozawa14}. Moreover, as the value of $n$ increases, the speed advantage of Algorithm 2 becomes more pronounced.
\begin{figure}[t]
\begin{center}
\includegraphics[width=0.6\linewidth]{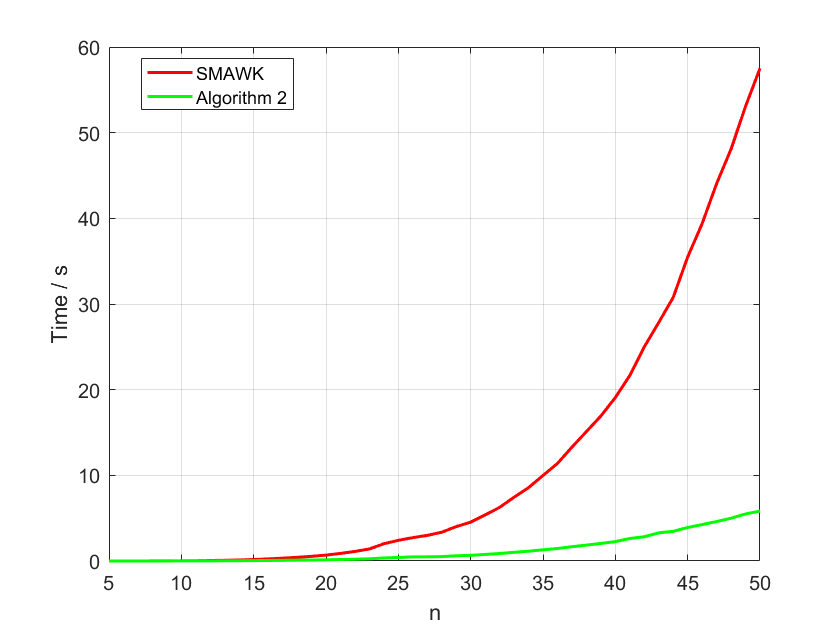}
\caption{{ \small Runtime comparison of Algorithm 2 and SMAWK$^{\cite{Ozawa14}}$}.}
 \label{fig04}
 \end{center}
\end{figure}

It should be noted that the processing of matrix $M_{j}$ in Step 4 neither skips the maximum value of each row nor alters the total monotonicity of the matrix. If $(M_{j})_{(a,b)}$ is not computed, it indicates that the corresponding $k_{1}$ and $k_{2}$ must violate the conditions in Theorem \ref{the05}. In this case, $(M_{j})_{(a,b)}$ can be regarded as the symmetric capacity corresponding to a certain degradation from $\sum_{i'\in[j+a]}q_{i'}\mathrm{B}(\sigma_{i'})$ is degraded to $\sum_{j'\in[j]}p_{j'}\mathrm{B}(\varepsilon_{j'})$. If this value were the maximum in row $a$, then $k_{1}, k_{2},\ldots, k_{j}$ would have to satisfy the conditions in Theorem \ref{the05}, which leads to a contradiction. Therefore, $(M_{j})_{(a,b)}$ cannot be the maximum value of row $a$. As for the total monotonicity of the matrix, since Step 4 only skips unnecessary computations without changing the structure of the matrix, the matrix processed by Step 4 still satisfies the definition of total monotonicity\cite[Definition 1]{Ozawa14}.

\section{Conclusions}
\label{sec08}
This paper presents a comprehensive study on the optimal degradations of symmetric binary-input discrete memoryless channels (BIDMCs), a key problem for polar code construction given the exponential growth of the output alphabet of synthetic channels. We fully characterize symmetric degradations of arbitrary symmetric BIDMCs and prove the invariance of the degradation relationship under the addition/removal of identical channel components. We identify all P-degradations that minimize decoding error probability, derive a criterion for (4,4)-P-degradations, and show every P-degradation is a degradation of a segment-structured P$^{\ast}$-degradation (requiring at least three splitting adjustments for upgradation). For C-optimal degradations that maximize symmetric capacity, we establish they must be P$^{+}$-degradations (special P$^{\ast}$-degradations) and derive a necessary condition via the function $\phi(\varepsilon_{1}, \varepsilon_{2})$, also proposing an algorithm to generate all C-optimal degradation candidates. Simulation results show that the number of C-optimal degradation candidates is far fewer than that of all possible P$^{+}$-degradations.

Finally, we propose an efficient algorithm for finding degradations with maximum symmetric capacity, and simulation results demonstrate that this method outperforms existing algorithms.

\section*{Acknowledgements}
This work was supported by the National Natural Science Foundation of China (No. 61977056). The authors would like to thank the anonymous reviewers and the editors for their valuable comments and constructive suggestions, which significantly improved the quality of this manuscript.

\section{Appendices}
\label{sec09}
\subsection{Appendix A: Proof of Theorem \ref{the01}}
\label{sec0901}
\begin{proof}
Assume that $W$ is a degradation of $Q$ with intermediate channel $P$. With respect to (\ref{04}), we see that, by ordering the outputs properly, the transition probability matrices of $W, Q, P$ can be written as
\[
W = \left( p_1 \begin{pmatrix} \overline{\varepsilon_1} & \varepsilon_1 \\ \varepsilon_1 & \overline{\varepsilon_1} \end{pmatrix} \; \cdots \; p_n \begin{pmatrix} \overline{\varepsilon_n} & \varepsilon_n \\ \varepsilon_n & \overline{\varepsilon_n} \end{pmatrix} \right),
\]
\[
Q = \left( q_1 \begin{pmatrix} \overline{\sigma_1} & \sigma_1 \\ \sigma_1 & \overline{\sigma_1} \end{pmatrix} \; \cdots \; q_m \begin{pmatrix} \overline{\sigma_m} & \sigma_m \\ \sigma_m & \overline{\sigma_m} \end{pmatrix} \right),
\]
\[
P =  \begin{pmatrix} \begin{pmatrix} a_{1,1} & b_{1,1} \\ c_{1,1} & d_{1,1} \end{pmatrix} & \cdots & \begin{pmatrix} a_{1,n} & b_{1,n} \\ c_{1,n} & d_{1,n} \end{pmatrix} \\ \vdots & \ddots & \vdots \\ \begin{pmatrix} a_{m,1} & b_{m,1} \\ c_{m,1} & d_{m,1} \end{pmatrix} & \cdots & \begin{pmatrix} a_{m,n} & b_{m,n} \\ c_{m,n} & d_{m,n} \end{pmatrix} \end{pmatrix} ,
\]
respectively, and satisfy the equation $QP = W$, that is, for each $j \in [n]$,
\[
p_j \begin{pmatrix} \overline{\varepsilon_j} & \varepsilon_j \\ \varepsilon_j & \overline{\varepsilon_j} \end{pmatrix} = \sum_{i \in [m]} q_i \begin{pmatrix} \overline{\sigma_i} & \sigma_i \\ \sigma_i & \overline{\sigma_i} \end{pmatrix} \begin{pmatrix} a_{i,j} & b_{i,j} \\ c_{i,j} & d_{i,j} \end{pmatrix},
\]
i.e.,
\begin{align}
\sum_{i \in [m]} q_i (\overline{\sigma_i} a_{i,j} + \sigma_i c_{i,j}) = \sum_{i \in [m]} q_i (\sigma_i b_{i,j} + \overline{\sigma_i} d_{i,j}) = p_j \overline{\varepsilon_j}, \label{41}\\
\sum_{i \in [m]} q_i (\overline{\sigma_i} b_{i,j} + \sigma_i d_{i,j}) = \sum_{i \in [m]} q_i (\sigma_i a_{i,j} + \overline{\sigma_i} c_{i,j}) = p_j \varepsilon_j. \label{42}
\end{align}
Notice that $a_{i,j}, b_{i,j}, c_{i,j}, d_{i,j}$ are nonnegative numbers satisfying
\[
\sum_{j \in [n]} (a_{i,j} + b_{i,j}) = \sum_{j \in [n]} (c_{i,j} + d_{i,j}) = 1, \; i \in [m].
\]

For $i \in [m]$ and $j \in [n]$, let $e_{i,j}$ be a number in $[0,1]$ with
\begin{align}\label{43}
2 p_{i,j} e_{i,j} = b_{i,j} + c_{i,j},
\end{align}
where $p_{i,j} = (a_{i,j} + b_{i,j} + c_{i,j} + d_{i,j})/2$. Clearly, $(p_{i,j})_{m \times n}$ is a conditional probability matrix. Furthermore, from (\ref{41}), (\ref{42}), (\ref{43}) and
\[
2 p_{i,j} \overline{e_{i,j}} = a_{i,j} + d_{i,j},
\]
for each $j \in [n]$, we have $\sum_{i \in [m]} q_i p_{i,j} = p_j$ and
\begin{align}
\sum_{i \in [m]} q_i p_{i,j} (\overline{\sigma_i}\ \overline{e_{i,j}} + \sigma_i e_{i,j}) = p_j \overline{\varepsilon_j},\nonumber\\
\sum_{i \in [m]} q_i p_{i,j} (\sigma_i \overline{e_{i,j}} + \overline{\sigma_i} e_{i,j}) = p_j \varepsilon_j. \label{44}
\end{align}
Therefore, we have
\[
p_j \begin{pmatrix} \overline{\varepsilon_j} & \varepsilon_j \\ \varepsilon_j & \overline{\varepsilon_j} \end{pmatrix} = \sum_{i \in [m]} q_i p_{i,j} \begin{pmatrix} \overline{\sigma_i} & \sigma_i \\ \sigma_i & \overline{\sigma_i} \end{pmatrix} \begin{pmatrix} \overline{e_{i,j}} & e_{i,j} \\ e_{i,j} & \overline{e_{i,j}} \end{pmatrix},
\]
and thus $\mathrm{B}(\varepsilon_j)$ is a degradation of $\sum_{i \in [m]} \frac{q_i p_{i,j}}{p_j} \mathrm{B}(\sigma_i)$ if $p_j > 0$. According to $\varepsilon_j, \sigma_1, \dots, \sigma_m \in [0, 1/2]$, we see that there are numbers $e_{1,j}, \dots, e_{m,j} \in [0,1]$ satisfying (\ref{44}) if and only if $p_j \varepsilon_j \geq \sum_{i \in [m]} q_i p_{i,j} \sigma_i$. Hence, we see this theorem is true by setting $k_{i,j} = q_i p_{i,j}$ for $i \in [m]$ and $j \in [n]$.
\end{proof}
\subsection{Appendix B: Proof of Lemma 6}
\begin{proof}
We give a proof only for the first conclusion. The other two conclusions can be proved similarly.

Since for any $i$ we have $\sum_{j \in [m]} k'_{i,j} = q_i$, according to Theorem \ref{the02} we see easily that $W'$ is a P-degradation of $Q$. If $p'_{j'} = 0$, then we must have $p_{j'} = k_{i',j'}, \varepsilon_{j'} = \sigma_{i'} = \varepsilon_{i'+1}$ and thus $W' \cong W$. Assume $p'_{j'} \neq 0$ now. Let
\[
R = \frac{p_{j'}}{p'} \mathrm{B}(\varepsilon_{j'}) + \frac{p_{j'+1}}{p'} \mathrm{B}(\varepsilon_{j'+1}), \quad R' = \frac{p'_{j'}}{p'} \mathrm{B}(\varepsilon'_{j'}) + \frac{p'_{j'+1}}{p'} \mathrm{B}(\varepsilon'_{j'+1}),
\]
where $p' = p_{j'} + p_{j'+1} = p'_{j'} + p'_{j'+1}$. From $k_{i',j'} \neq 0$ and $\sigma_{i'} \geq \varepsilon_{j'+1} \geq \varepsilon_{j'}$, we see
\[
\varepsilon'_{j'} = \frac{\sum_{i \in [m]} k'_{i,j'} \sigma_i}{p'_{j'}} = \frac{\sum_{i \in [m] \setminus \{i'\}} k_{i,j'} \sigma_i}{ \sum_{i \in [m] \setminus \{i'\}} k_{i,j'} } = \frac{p_{j'} \varepsilon_{j'} - k_{i',j'} \sigma_{i'}}{p_{j'} - k_{i',j'}} \leq \varepsilon_{j'},
\]
\[
\varepsilon'_{j'+1} = \frac{\sum_{i \in [m]} k'_{i,j'+1} \sigma_i}{p'_{j'+1}} = \frac{k_{i',j'} \sigma_{i'} + \sum_{i \in [m]} k_{i,j'+1} \sigma_i}{k_{i',j'} + \sum_{i \in [m]} k_{i,j'+1}} = \frac{k_{i',j'} \sigma_{i'} + p_{j'+1} \varepsilon_{j'+1}}{k_{i',j'} + p_{j'+1}} \geq \varepsilon_{j'+1}.
\]
Therefore, according to
\begin{align*}
\frac{p'_{j'}}{p'} \varepsilon'_{j'} + \frac{p'_{j'}}{p'} \varepsilon'_{j'+1} &= \frac{1}{p'} \sum_{i \in [m]} (k'_{i,j'} + k'_{i,j'+1}) \sigma_i\\ &= \frac{1}{p'} \sum_{i \in [m]} (k_{i,j'} + k_{i,j'+1}) \sigma_i = \frac{p_{j'}}{p'} \varepsilon_{j'} + \frac{p_{j'}}{p'} \varepsilon_{j'+1}
\end{align*}
and Lemma \ref{lem05}, we see $R \preccurlyeq R'$. Let $T = \sum_{j \in [n] \setminus \{j',j'+1\}} \frac{p_j}{1-p'} \mathrm{B}(\varepsilon_j)$. Then, according to Corollary \ref{cor02} we have
\[
W \cong p' R + (1-p') T \preccurlyeq p' R' + (1-p') T \cong W'.
\]
The proof is complete.
\end{proof}

\subsection{Appendix C: Proof of Theorem 4}
To prove Theorem \ref{the04}, we assume
\begin{align}
W = D_Q(i_2, \dots, i_n; s_2, \dots, s_n) = \sum_{j \in [n]} p_j \mathrm{B}(\varepsilon_j), \label{45}\\
W' = D_Q(i'_2, \dots, i'_n; s'_2, \dots, s'_n) = \sum_{j \in [n]} p'_j \mathrm{B}(\varepsilon'_j), \label{46}
\end{align}
are $2n$-P$^{\ast}$-degradations of $Q = \sum_{i \in [m]} q_i \mathrm{B}(\sigma_i) \in \mathbb{B}_m^*$, where $0 \leq \sigma_1 < \cdots < \sigma_m \leq 1/2$, $0 \leq \varepsilon_1 < \cdots < \varepsilon_n \leq 1/2$ and $0 \leq \varepsilon'_1 < \cdots < \varepsilon'_n \leq 1/2$. We prove the following lemma at first.
\begin{lemma}\label{lem09}
Assume that there are integers $k, j$ such that $(i_k, s_k) = (i'_j, s'_j)$. Then, we have $p_k \geq p'_j$ if $\varepsilon_k \geq \varepsilon'_j$, and $p_{k-1} \geq p'_{j-1}$ if $\varepsilon_{k-1} \leq \varepsilon'_{j-1}$.
\end{lemma}
\begin{proof}
Assume $\varepsilon_k \geq \varepsilon'_j$ and $p_k < p'_j$. Without loss of generality, we assume further that $q_{i'_{j+1}} - s'_{j+1} > 0$. From $p_k < p'_j$ and $(i_k, s_k) = (i'_j, s'_j)$ we see $i_{k+1} \leq i'_{j+1}$. Since $\varepsilon_k$ is a weighted average of the numbers $\sigma_{i_k}, \dots, \sigma_{i_{k+1}}$ and $\delta = (p'_j \varepsilon'_j - p_k \varepsilon_k)/(p'_j - p_k)$ is a weighted average of the numbers $\sigma_{i_{k+1}}, \dots, \sigma_{i'_{j+1}}$, from $\varepsilon_k \geq \varepsilon'_j \geq \delta$ we see $s_k = 0$ and $\sigma_{i_{k}+1} = \cdots = \sigma_{i_{k+1}} = \cdots = \sigma_{i'_{j+1}}$. Hence, we have $i'_{j+1} = i_{k+1} = i_k + 1$, $\varepsilon_k = \varepsilon'_j = \sigma_{i_{k+1}}$ and thus $p_k < p'_j \leq q_{i_{k+1}}$, contradicts that $W$ is a $2n$-P$^{\ast}$-degradations of $Q$. Hence, we have $p_k \geq p'_j$ if $\varepsilon_k \geq \varepsilon'_j$.

Assume $\varepsilon_{k-1} \leq \varepsilon'_{j-1}$ and $p_{k-1} < p'_{j-1}$ now. Without loss of generality, we assume further that $s'_{j-1} > 0$. From $p_{k-1} < p'_{j-1}$ and $(i_k, s_k) = (i'_j, s'_j)$ we see $i_{k-1} \geq i'_{j-1}$. Since $\varepsilon_{k-1}$ is a weighted average of the numbers $\sigma_{i_{k-1}}, \dots, \sigma_{i_k}$ and $\delta' = (p'_{j-1} \varepsilon'_{j-1} - p_{k-1} \varepsilon_{k-1})/(p'_{j-1} - p_{k-1})$ is a weighted average of the numbers $\sigma_{i'_{j-1}}, \dots, \sigma_{i_{k-1}}$, from $\varepsilon_{k-1} \leq \varepsilon'_{j-1} \leq \delta'$ we see $s_k = q_{i_k}$ and $\sigma_{i'_{j-1}} = \cdots = \sigma_{i_{k-1}} = \cdots = \sigma_{i_{k-1}}$. Hence, we have $i'_{j-1} = i_{k-1} = i_k - 1$, $\varepsilon_{k-1} = \varepsilon'_{j-1} = \sigma_{i_{k-1}}$ and thus $p_{k-1} < p'_{j-1} \leq q_{i_{k-1}}$, contradicts that $W$ is a $2n$-P$^{\ast}$-degradations of $Q$. Hence, we have $p_{k-1} \geq p'_{j-1}$ if $\varepsilon_{k-1} \leq \varepsilon'_{j-1}$.
\end{proof}
Now we give a proof for Theorem \ref{the04}.
\begin{proof}
Assume, in contrast, that the channels (\ref{45}) and (\ref{46}) are distinct $2n$-P$^{\ast}$-degradations of $Q$ such that $W \preccurlyeq W'$ and their splitting patterns differ at less than three positions.

For the first case that the splitting patterns of $W$ and $W'$ differ at just one position, let $l$ be the unique number with $(i_j, s_j) \neq (i'_j, s'_j)$. According to Theorem \ref{the03} and Corollary \ref{cor02}, we have $p = p_{j-1} + p_j = p'_{j-1} + p'_j$ and $\frac{p_{j-1}}{p} \mathrm{B}(\varepsilon_{j-1}) + \frac{p_j}{p} \mathrm{B}(\varepsilon_j) \preccurlyeq_{\mathrm{P}} \frac{p'_{j-1}}{p} \mathrm{B}(\varepsilon'_{j-1}) + \frac{p'_j}{p} \mathrm{B}(\varepsilon'_j)$. Then, from Lemma \ref{lem04} we see $\varepsilon'_{j-1} \leq \varepsilon_{j-1}$ and $\varepsilon'_j \geq \varepsilon_j$. According to Lemma \ref{lem09}, from $(i_{j-1}, s_{j-1}) = (i'_{j-1}, s'_{j-1})$ and $\varepsilon'_{j-1} \leq \varepsilon_{j-1}$ we see $p'_{j-1} \leq p_{j-1}$, from $(i_{j+1}, s_{j+1}) = (i'_{j+1}, s'_{j+1})$ and $\varepsilon'_j \geq \varepsilon_j$ we see $p'_j \leq p_j$. Therefore, from $p_{j-1} + p_j = p'_{j-1} + p'_j$ we see $p_{j-1} = p'_{j-1}, p_j = p'_j$ and $W \cong W'$, contradicts our assumption.

For the second case that the splitting patterns of $W, W'$ differ at two consecutive positions, let $j$ be the unique number such that $(i_j, s_j) \neq (i'_j, s'_j)$ and $(i_{j+1}, s_{j+1}) \neq (i'_{j+1}, s'_{j+1})$. From Theorem \ref{the03} and Corollary \ref{cor02} we see $p = p_{j-1} + p_j + p_{j+1} = p'_{j-1} + p'_j + p'_{j+1}$ and
\begin{align}\label{47}
\sum_{l=1}^3 \frac{p_{l+j-2}}{p} \mathrm{B}(\varepsilon_{l+j-2}) \preccurlyeq_{\mathrm{P}} \sum_{l=1}^3 \frac{p'_{l+j-2}}{p} \mathrm{B}(\varepsilon'_{l+j-2}).
\end{align}
Hence, from Lemma \ref{lem04} we have $\varepsilon'_{j-1} \leq \varepsilon_{j-1}$ and $\varepsilon'_{j+1} \geq \varepsilon_{j+1}$. According to Lemma \ref{lem09}, one can show easily $x = p_{j-1} - p'_{j-1} > 0$ and $y = p_{j+1} - p'_{j+1} > 0$. Hence, we have $i'_j \leq i_j$ and $i_{j+1} \leq i'_{j+1}$. Let $\delta_1 = (p_{j-1} \varepsilon_{j-1} - p'_{j-1} \varepsilon'_{j-1})/x$ and $\delta_2 = (p_{j+1} \varepsilon_{j+1} - p'_{j+1} \varepsilon'_{j+1})/y$. Since $\delta_1$ is a weighted average of the numbers $\sigma_{i'_j}, \dots, \sigma_{i_j}$ and $\delta_2$ is a weighted average of the numbers $\sigma_{i_{j+1}}, \dots, \sigma_{i'_{j+1}}$, we have $0 \leq \varepsilon_{j-1} \leq \delta_1 \leq \varepsilon_j \leq \delta_2 \leq \varepsilon_{j+1} \leq \frac{1}{2}$ and $\delta_1 < \delta_2$. According to Theorem \ref{the02}, there is a 1-matrix $(r_{i,j})_{3 \times 3}$ such that
\begin{align}
p(r_{1,1} + r_{1,2} + r_{1,3}) &= p'_{j-1} = p_{j-1} - x, \label{48}\\
p(r_{2,1} + r_{2,2} + r_{2,3}) &= p'_j = p_j + x + y, \label{49}\\
p(r_{3,1} + r_{3,2} + r_{3,3}) &= p'_{j+1} = p_{j+1} - y, \label{50}\\
p(r_{1,l} + r_{2,l} + r_{3,l}) &= p_{j+l-2}, \; l = 1,2,3, \label{51}\\
p(r_{1,l} \varepsilon'_{j-1} + r_{2,l} \varepsilon'_j + r_{3,l} \varepsilon'_{j+1}) &= \varepsilon_{j+l-2} p_{j+l-2}, \; l = 1,2,3. \label{52}
\end{align}
Then from (\ref{48}), (\ref{49}) and (\ref{51}) we have
\begin{align*}
&p(r_{1,2} + r_{1,3} - r_{3,1} + r_{2,2} + r_{2,3})\\ =& (p_{j-1} - x - p r_{1,1}) - p r_{3,1} + (p_j + x + y - p r_{2,1}) = p_j + y
\end{align*}
and thus from
\begin{align*}
p'_j \varepsilon'_j &= p_{j-1} \varepsilon_{j-1} + p_j \varepsilon_j + p_{j+1} \varepsilon_{j+1} - p'_{j-1} \varepsilon'_{j-1} - p'_{j+1} \varepsilon'_{j+1}\\ &= p_j \varepsilon_j + x \delta_1 + y \delta_2,
\end{align*}
(\ref{50}) and (\ref{52}) we see
\begin{align*}
&p(r_{1,2} + r_{1,3})(\varepsilon'_{j-1} - \varepsilon'_j) + p r_{3,1}(\varepsilon'_j - \varepsilon'_{j+1})\\
=& p_j \varepsilon_j + p_{j+1} \varepsilon_{j+1} - p(r_{3,1} + r_{3,2} + r_{3,3}) \varepsilon'_{j+1} - p(r_{1,2} + r_{1,3} - r_{3,1} + r_{2,2} + r_{2,3}) \varepsilon'_j\\
=&p_{j}\varepsilon_{j}+p_{j+1}\varepsilon_{j+1}-p'_{j+1}\varepsilon'_{j+1}-(p_{j}+y)\varepsilon'_{j}\\
=& p_j \varepsilon_j + p_{j+1} \varepsilon_{j+1} - (p_{j+1}\varepsilon_{j+1} - y \delta_2) - \frac{(p_j + y)(p_j \varepsilon_j + x \delta_1 + y \delta_2)}{p_j + x + y}\\
=& \frac{xy(\delta_2 - \delta_1) + p_j x(\varepsilon_j - \delta_1)}{p_j + x + y} > 0,
\end{align*}
contradicts to $\varepsilon'_{j-1} < \varepsilon'_j < \varepsilon'_{j+1}$.

For the last case that the splitting patterns of $W, W'$ differ at two separated positions, assume that $j, k$ are the integers such that $2 \leq j < k \leq n-1$, $(i_j, s_j) \neq (i'_j, s'_j)$ and $(i_{k+1}, s_{k+1}) \neq (i'_{k+1}, s'_{k+1})$. According to Theorem \ref{the03} and Corollary \ref{cor02} we see
\[
p_{j-1} + p_j = p'_{j-1} + p'_j, \; p_{j-1} \varepsilon_{j-1} + p_j \varepsilon_j = p'_{j-1} \varepsilon'_{j-1} + p'_j \varepsilon'_j,
\]
\[
p_k + p_{k+1} = p'_k + p'_{k+1}, \; p_k \varepsilon_k + p_{k+1} \varepsilon_{k+1} = p'_k \varepsilon'_k + p'_{k+1} \varepsilon'_{k+1},
\]
and that
\[
\frac{p_{j-1}}{p} \mathrm{B}(\varepsilon_{j-1}) + \frac{p_j}{p} \mathrm{B}(\varepsilon_j) + \frac{p_k}{p} \mathrm{B}(\varepsilon_k) + \frac{p_{k+1}}{p} \mathrm{B}(\varepsilon_{k+1})
\]
is a P-degradation of
\[
\frac{p'_{j-1}}{p} \mathrm{B}(\varepsilon'_{j-1}) + \frac{p'_j}{p} \mathrm{B}(\varepsilon'_j) + \frac{p'_k}{p} \mathrm{B}(\varepsilon'_k) + \frac{p'_{k+1}}{p} \mathrm{B}(\varepsilon'_{k+1}).
\]
Then, from Theorem \ref{the02}, there is a 1-matrix $(r_{i,j})_{4 \times 4}$ such that
\[
r_{1,1} + r_{1,2} + r_{1,3} + r_{1,4} = p'_{j-1}/p,
\]
\[
r_{2,1} + r_{2,2} + r_{2,3} + r_{2,4} = p'_j/p,
\]
\[
r_{3,1} + r_{3,2} + r_{3,3} + r_{3,4} = p'_k/p,
\]
\[
r_{4,1} + r_{4,2} + r_{4,3} + r_{4,4} = p'_{k+1}/p,
\]
\[
r_{1,1} + r_{2,1} + r_{3,1} + r_{4,1} = p_{j-1}/p,
\]
\[
r_{1,2} + r_{2,2} + r_{3,2} + r_{4,2} = p_j/p,
\]
\[
r_{1,3} + r_{2,3} + r_{3,3} + r_{4,3} = p_k/p,
\]
\[
r_{1,4} + r_{2,4} + r_{3,4} + r_{4,4} = p_{k+1}/p,
\]
\[
r_{1,1} \varepsilon'_{j-1} + r_{2,1} \varepsilon'_j + r_{3,1} \varepsilon'_k + r_{4,1} \varepsilon'_{k+1} = p_{j-1} \varepsilon_{j-1}/p,
\]
\[
r_{1,2} \varepsilon'_{j-1} + r_{2,2} \varepsilon'_j + r_{3,2} \varepsilon'_k + r_{4,2} \varepsilon'_{k+1} = p_j \varepsilon_j/p,
\]
\[
r_{1,3} \varepsilon'_{j-1} + r_{2,3} \varepsilon'_j + r_{3,3} \varepsilon'_k + r_{4,3} \varepsilon'_{k+1} = p_k \varepsilon_k/p,
\]
\[
r_{1,4} \varepsilon'_{j-1} + r_{2,4} \varepsilon'_j + r_{3,4} \varepsilon'_k + r_{4,4} \varepsilon'_{k+1} = p_{k+1} \varepsilon_{k+1}/p.
\]
Then, from
\begin{align*}
&(r_{1,1} + r_{2,1} + r_{3,1} + r_{4,1}) + (r_{1,2} + r_{2,2} + r_{3,2} + r_{4,2})\\
=& p_{j-1}/p + p_j/p = p'_{j-1}/p + p'_j/p\\
=& (r_{1,1} + r_{1,2} + r_{1,3} + r_{1,4}) + (r_{2,1} + r_{2,2} + r_{2,3} + r_{2,4})
\end{align*}
we see
\begin{align}\label{53}
(r_{3,1} + r_{3,2}) + (r_{4,1} + r_{4,2}) = (r_{1,3} + r_{1,4}) + (r_{2,3} + r_{2,4}).
\end{align}
From
\begin{align*}
&(r_{1,1} + r_{1,2}) \varepsilon'_{j-1} + (r_{2,1} + r_{2,2}) \varepsilon'_j + (r_{3,1} + r_{3,2}) \varepsilon'_k + (r_{4,1} + r_{4,2}) \varepsilon'_{k+1}\\
=& (p_{j-1} \varepsilon_{j-1} + p_j \varepsilon_j)/p = (p'_{j-1} \varepsilon'_{j-1} + p'_j \varepsilon'_j)/p\\
=& (r_{1,1} + r_{1,2} + r_{1,3} + r_{1,4}) \varepsilon'_{j-1} + (r_{2,1} + r_{2,2} + r_{2,3} + r_{2,4}) \varepsilon'_j
\end{align*}
we see
\begin{align}\label{54}
(r_{3,1} + r_{3,2}) \varepsilon'_k + (r_{4,1} + r_{4,2}) \varepsilon'_{k+1} = (r_{1,3} + r_{1,4}) \varepsilon'_{j-1} + (r_{2,3} + r_{2,4}) \varepsilon'_j.
\end{align}
From (\ref{53}), (\ref{54}) and $\varepsilon'_{j-1} < \varepsilon'_j < \varepsilon'_k < \varepsilon'_{k+1}$, we see
\[
r_{3,1} + r_{3,2} = r_{4,1} + r_{4,2} = r_{1,3} + r_{1,4} = r_{2,3} + r_{2,4} = 0.
\]
Hence, we have
\[
p_{j-1} \mathrm{B}(\varepsilon_{j-1}) + p_j \mathrm{B}(\varepsilon_j) \preccurlyeq_{\mathrm{P}} p'_{j-1} \mathrm{B}(\varepsilon'_{j-1}) + p'_j \mathrm{B}(\varepsilon'_j),
\]
which implies that from $W$ one can obtain a new $2n$-P$^{\ast}$-degradation $W''$ of $Q$ by replacing the splitting pattern $(i_j, s_j)$ with $(i'_j, s'_j)$ such that $W \preccurlyeq W''$, this is impossible.
\end{proof}

\subsection{Appendix D: Proof of Lemma 7}
Clearly, we have the following lemma.
\begin{lemma}\label{lem10}
Let $f(x)$ be a function whose second derivative $f''(x)$ is positive on the interval $(a,b)$. Then for $a < c < b$ the function $f(x)$ is
\begin{enumerate}
\item decreasing on $(a,c]$ if $f'(c) \leq 0$,
\item increasing on $[c,b)$ if $f'(c) \geq 0$.
\end{enumerate}
\end{lemma}

Now we give a proof for Lemma \ref{lem07}.
\begin{proof}
For $a, t \in (0, 1/2)$, let
\[
g_a(t) = a \ln \frac{\overline{t}a}{\overline{a}t} - \ln \frac{\overline{t}}{\overline{a}}.
\]
Then, we have $g'_a(t) = (t - a)/(\overline{t}t)$ and thus $g_a(t)$ is decreasing on $(0,a]$ and increasing on $[a, 1/2)$. Hence, we have $g_a(t) > g_a(a) = 0$ for $t \in (0, 1/2) \setminus \{a\}$. Then, from $\ln((\overline{\varepsilon_1} \varepsilon_2)/(\overline{\varepsilon_2} \varepsilon_1)) > 0$ we see
\[
\phi(\varepsilon_1, \varepsilon_2) - \varepsilon_2 = \frac{\ln(\overline{\varepsilon_1}/\overline{\varepsilon_2}) - \varepsilon_2 \ln((\overline{\varepsilon_1} \varepsilon_2)/(\overline{\varepsilon_2} \varepsilon_1))}{\ln((\overline{\varepsilon_1} \varepsilon_2)/(\overline{\varepsilon_2} \varepsilon_1))} = \frac{-g_{\varepsilon_2}(\varepsilon_1)}{\ln((\overline{\varepsilon_1} \varepsilon_2)/(\overline{\varepsilon_2} \varepsilon_1))} < 0,
\]
\[
\phi(\varepsilon_1, \varepsilon_2) - \varepsilon_1 = \frac{\ln(\overline{\varepsilon_1}/\overline{\varepsilon_2}) - \varepsilon_1 \ln((\overline{\varepsilon_1} \varepsilon_2)/(\overline{\varepsilon_2} \varepsilon_1))}{\ln((\overline{\varepsilon_1} \varepsilon_2)/(\overline{\varepsilon_2} \varepsilon_1))} = \frac{g_{\varepsilon_1}(\varepsilon_2)}{\ln((\overline{\varepsilon_1} \varepsilon_2)/(\overline{\varepsilon_2} \varepsilon_1))} > 0,
\]
and thus (\ref{35}) follows.

According to $\varepsilon_1 \leq \sigma \leq \varepsilon_2$ and $p_1 + p_2 \leq 1$ we see the joint set of $(-p_2 b_{\sigma, \varepsilon_2}, 1 - p_2)$ and $(p_1 - 1, p_1 b_{\sigma, \varepsilon_1})$ is
\[
(-p_2(1 - 2 \varepsilon_2)/(1 - 2 \sigma), p_1 \varepsilon_1/\sigma) = (-p_2 b_{\sigma, \varepsilon_2}, p_1 b_{\sigma, \varepsilon_1})
\]
and thus (\ref{36}) follows from
\[
e_{\sigma, \varepsilon_1, p_1}(-x) + \frac{p_1(\sigma - \varepsilon_1)}{p_1 - x} = \sigma = e_{\sigma, \varepsilon_2, p_2}(x) - \frac{p_2(\varepsilon_2 - \sigma)}{p_2 + x}.
\]

From (\ref{31}) and (\ref{32}), we see that $F''_{\sigma, \varepsilon_1, p_1, \varepsilon_2, p_2}(x)$ is positive for each $x \in (-p_2 b_{\sigma, \varepsilon_2}, p_1 b_{\sigma, \varepsilon_1})$. If $\sigma \leq \phi(\varepsilon_1, \varepsilon_2)$, from (\ref{33}) we see $F'_{\sigma, \varepsilon_1, p_1, \varepsilon_2, p_2}(0) \leq 0$ and, according to Lemma \ref{lem10}, $F_{\sigma, \varepsilon_1, p_1, \varepsilon_2, p_2}(x)$ is decreasing on $(-p_2 b_{\sigma, \varepsilon_2}, 0]$. If $\sigma \geq \phi(\varepsilon_1, \varepsilon_2)$, from (\ref{33}) we see $F'_{\sigma, \varepsilon_1, p_1, \varepsilon_2, p_2}(0) \geq 0$ and, according to Lemma \ref{lem10}, $F_{\sigma, \varepsilon_1, p_1, \varepsilon_2, p_2}(x)$ is increasing on $[0, p_1 b_{\sigma, \varepsilon_1})$.
\end{proof}
\end{document}